%% file: adGarch_2009_2.tex
\documentclass[a4,12pt,reqno]{article}
\usepackage{harvard}
\usepackage{graphicx}
\usepackage{amsmath,amssymb,amsthm}

\usepackage[margin=25mm, a4paper, dvips]{geometry}
\usepackage{comment}
\usepackage{epsfig}

\input{expdef}
\input{mydef}

\newcommand{\keywords}[1]{\par\noindent\emph{Keywords:} #1 \\}

\def\ko{k^{*}}

\numberwithin{equation}{section}
\numberwithin{figure}{section}
\newcounter{example}[section]
\numberwithin{example}{section}
\newcounter{remark}[section]
\numberwithin{remark}{section}
\newtheorem{theorem}{Theorem}[section]

\newtheorem{lemma}[theorem]{Lemma}
\newtheorem{corollary}[theorem]{Corollary}

\newtheorem{exmp}[example]{Example}
\newtheorem{rmrk}[remark]{Remark}
\newenvironment{example}{\begin{exmp}\rm}{\end{exmp}}
\newenvironment{remark}{\begin{rmrk}\rm}{\end{rmrk}}
\renewenvironment{abstract}
    {
      \begin{center}
       \begin{minipage}{16cm}
        \begin{small}
    }
    {   \end{small}
       \end{minipage}
      \end{center}
     \bigskip
    }

\def\TT{T}
\def\eps{\varepsilon}
\def\tcp{\tau}

\def\Tcp{\cc{T}}

\def\IK{\cc{I}}

\begin{document}
\date{\today}

\title
{Adaptive pointwise estimation in time-inhomogeneous conditional heteroscedasticity models}

\author{
P. \v{C}\'\i\v{z}ek,\footnote{Dept.\ of Econometrics \& OR, Tilburg University, P.O.Box 
90153, 5000LE Tilburg, \vspace{-1.6mm}The Netherlands.} 
~W. H\"{a}rdle,\footnote{Humboldt-Universit\"at zu Berlin, Spandauerstrass{}e 1, 10178 
Berlin, Germany.} 
~and 
V. Spokoiny\footnote{Weierstrass-Institute, Mohrenstr. 39, 10117 Berlin, Germany.}
}
\date{}

\maketitle
\pagestyle{plain}
\markboth%
{\hfil {\sc \small inference for time-inhomogeneous models}
\hfil}%
{\hfil {\sc \small p. \v{c}\'\i\v{z}ek, w. h\"{a}rdle, and v. spokoiny} \hfil}


\begin{abstract}
{\normalfont\small
This paper offers a new method for estimation and forecasting of 
the volatility of financial
time series when the stationarity assumption is violated.
Our general local parametric approach particularly applies to
general varying-coefficient parametric models, such as GARCH, whose
coefficients may arbitrarily vary with time.
Global parametric, smooth transition, and change-point models are special cases.
The method is
based on an adaptive pointwise selection of the largest interval of homogeneity
with a given right-end point by a local change-point
analysis. We construct locally adaptive estimates that can perform this task
and investigate them both from the theoretical point of view and by Monte Carlo
simulations. In the particular case of GARCH estimation, the proposed method is
applied to stock-index series and is shown to outperform the standard parametric
GARCH model. 
}
\end{abstract}


%


\section{Introduction}\label{intro}

A growing amount of econometrical and statistical research is devoted to modeling
financial time series and their volatility, which measures
dispersion at a point in time (i.e., conditional variance).
Although many economies and financial markets have been recently experiencing
many shorter and longer periods of instability or uncertainty such as Asian
crisis (1997), Russian crisis (1998), start of the European currency (1999),
the ``dot-Com'' technology-bubble crash (2000--2002), or the terrorist attacks
(September, 2001), the war in Iraq (2003), and the current global recession
(2008), mostly used econometric models are based on the assumption of time
homogeneity. This includes linear and nonlinear autoregressive (AR) and
moving-average models and conditional heteroscedasticity (CH) models such as
ARCH (Engel, 1982) and GARCH (Bollerslev, 1986), stochastic volatility models
(Taylor, 1986), as well as their combinations such as AR-GARCH.

On the other hand, the market and institutional changes have long been assumed
to cause structural breaks in 
financial time series, which
was confirmed, for example, in data on
stock prices (Andreou and Ghysels, 2002;
Beltratti and Morana, 2004) and exchange rates (Herwatz and Reimers, 2001).
Moreover, ignoring these breaks can adversely affect the modeling, estimation,
and forecasting of volatility 
as suggested by Diebold and Inoue (2001), Mikosch and Starica (2004), Pesaran
and Timmermann (2004), and Hillebrand (2005), for instance. Such findings led
to the development of the change-point analysis in the context of 
CH models; see for example, Chen and Gupta (1997), 
Kokoszka and Leipus (2000), 
and Andreou and Ghysels (2006).

An alternative approach lies in relaxing the assumption of time homogeneity and
allowing some or all model parameters to vary over time 
(Chen and Tsay, 1993; Cai et al., 2000; Fan and Zhang, 2008). 
Without structural assumptions about the transition of model parameters over
time, time-varying coefficient models have to be estimated nonparametrically,
for example, under the identification condition that their parameters are
smooth functions of time (Cai et al., 2000). In this paper, we follow a
different strategy based on the assumption that a time series can be locally,
that is over short periods of time, approximated by a parametric model. As
suggested by Spokoiny (1998), such a local approximation can form a starting
point in the search for the longest period of stability (homogeneity), that is,
for the longest time interval in which the series is described well by the
parametric model. In the context of the local constant approximation, this
strategy was employed for volatility modeling by H\"ardle et al.\ (2003),
Mercurio and Spokoiny (2004), and Spokoiny (2008). Our aim is to generalize
this approach so that it can identify intervals of homogeneity for any
parametric CH model regardless of its complexity.

In contrast to the local constant approximation of the volatility of a process
(Mercurio and Spokoiny, 2004), the main benefit of the proposed generalization
consists in the possibility to apply the methodology to a much wider class of
models and to forecast over a longer time horizon.
The reason is that approximating the mean or volatility process by a constant
is in many cases too restrictive or even inappropriate and it is fulfilled only
for short time intervals, which precludes its use for longer-term forecasting.
On the contrary, parametric models like 
GARCH mimic the majority of stylized facts about 
financial time series and can reasonably fit the data over rather long periods
of time in many practical situations. Allowing for time dependence of model
parameters offers then much more flexibility in modeling real-life time series,
which can be both with or without structural breaks since global parametric
models are included as a special case.

Moreover, the proposed adaptive local parametric modeling unifies the
change-point and varying-coefficient models. First, since finding the longest
time-homogen\-eous interval for a parametric model at any point in time
corresponds to detecting the most recent change-point in a time series, this
approach resembles the change-point modeling as in Bai and Perron (1998)
or Mikosch and Starica (1999, 2004), for instance, but it does not require
prior information such as the number of changes.
Additionally, the traditional structural-change tests 
require that the number of observations before each break point
is large (and can grow to infinity) as these tests rely on asymptotic results. 
On the contrary, the proposed pointwise adaptive estimation does not rely on
asymptotic results and does not thus place any requirements on the number of
observations before, between, or after any break point. Second, since the
adaptively selected time-homogeneous interval used for estimation necessarily
differs at each time point, the model coefficients can arbitrarily vary over
time. In comparison to varying-coefficient models assuming smooth development
of parameters over time (Cai et al., 2000), our approach however allows for
structural breaks in the form of sudden jumps in parameter values.

Although seemingly straightforward, extending Mercurio and Spokoiny (2004)'s
procedure to the local parametric modeling is a nontrivial problem, which
requires new tools and techniques. We concentrate here on the change-point
estimation of financial time series, which are often modelled by data-demanding
models such as GARCH. While the benefits of a flexible change-point analysis
for time series spanning several years are well known, its feasibility (which
stands in the focus of this work) is much more difficult to achieve. The reason
is thus that, at each time point, the procedure starts from a small interval,
where a local parametric approximation holds, and then iteratively extends this
interval and tests it for time-homogeneity until a structural break is found or
data exhausted. Hence, a model has to be initially estimated on very short time
intervals (e.g., 10 observations). Using standard testing methods, such a
procedure might be feasible for simple parametric models, 
but it is hardly possible for more complex parametric models such as GARCH that 
generally require rather large samples for reasonably good estimates. 

Therefore, we use an alternative and more robust approach to local change-point analysis
that relies on a finite-sample theory of testing a growing sequence of
historical time intervals on homogeneity against a change-point alternative. 
The proposed adaptive pointwise estimation procedure applies to a wide
class of time-series models, including AR and CH models. Concentrating 
on the latter, we describe in details the adaptive procedure, 
derive its basic properties, 
and focusing on the feasibility of adaptive estimation for CH models,
study the performance in comparison to the parametric (G)ARCH by means of
simulations and real-data applications. The main conclusion is two-fold: on one
hand, the adaptive pointwise estimation is feasible and beneficial also in
the case of data-demanding models such as GARCH; on the other hand, the
adaptive estimates based on various parametric models such as constant, ARCH,
or GARCH models are much to closer to each other (while being better than the
usual parametric estimates), which eliminates to some extent the need for using
too complex models in adaptive estimation.

The rest of the paper is organized as follows. In Section~\ref{Sparamodel}, the
parametric estimation of CH models and its finite-sample properties are
introduced. In Section~\ref{Sprocedure}, we define the adaptive pointwise estimation
procedure and discuss the choice of its parameters. Theoretical properties of
the method are discussed in Section~\ref{Sresults}. In the specific case of the
ARCH(1) and GARCH(1,1) models, a simulation study illustrates the performance
of the new methodology with respect to the standard parametric and change-point
models in Section~\ref{Sappl}. Applications to real stock-index series data are
presented in Section~\ref{Srealappls}. The proofs are provided in the Appendix.

\section{Parametric conditional heteroscedasticity models} \label{Sparamodel}
\label{Sparametric}

Consider a time series \( Y_{t} \) in discrete time, \( t \in N \).
The conditional heteroscedasticity assumption means that \( Y_{t} = \sigma_{t} \eps_{t} \), where 
\( \{ \eps_{t} \}_{t \in N} \) is a white noise process and \( \{ \sigma_{t} \}_{t \in N} \) 
is a predictable volatility (conditional variance) process. 
Modelling of the volatility process \( \sigma_{t} \) typically relies on some parametric 
CH specification such as the ARCH (Engle, 1982) and GARCH (Bollerslev, 1986) models:
\begin{equation}
\label{garch_pq}
    \sigma_{t}^{2}
    =
    \omega + \sum_{i=1}^{p} \alpha_{i} Y_{t-i}^{2}
    + \sum_{j=1}^{q} \beta_{j} \sigma_{t-j}^{2},
\end{equation}
where \( p \in N \), \( q \in N \), and 
\( \thetav = (\omega, \alpha_1,\ldots,\alpha_{p}, \beta_1,\ldots,\beta_{q})^{\T} \)
is the parameter vector. An attractive feature of this model is that, even with 
very few coefficients, one can model most stylized facts of financial time series 
like volatility clustering or excessive kurtosis, for instance. 
A number of (G)ARCH extensions were proposed to make the model even more
flexible; for example, EGARCH (Nelson, 1991), QGARCH (Sentana, 1995), and
TGARCH (Glosten et al., 1993)
that account for asymmetries in a volatility process.

All such CH models can be put into a common class
of generalized linear volatility models:
\vspace{-2ex}
\begin{eqnarray}
    Y_{t} 
    &=& 
    \sigma_{t} \eps_{t} = \sqrt{g(X_{t})} \eps_{t}, 
    \label{Yspt}\label{Yproc}\label{timemodel}
    \\
    X_{t} 
    &=& 
    \omega + \sum_{i=1}^{p} \alpha_{i} h(Y_{t-i}) + \sum_{j=1}^{q} \beta_{j} X_{t-j} \, , 
    \label{gltm}\label{Xtparam}
\end{eqnarray}
where \( g \) and \( h \) are known functions and \( X_{t} \) is a (partially)
unobserved process (structural variable) that models the volatility coefficient 
\( \sigma_{t}^{2} \) via transformation \( g \): \( \sigma_{t}^{2} = g(X_{t}) \). 
For example, the GARCH model (\ref{garch_pq})
is described by \( g(u) = u \) and \( h(r) = r^{2} \). 

Model (\ref{timemodel})--(\ref{Xtparam}) is time homogeneous in the sense that
the process \( Y_{t} \) follows the same structural equation at each time
point. In other words, the parameter \(\thetav \) and hence the structural
dependence in \( Y_{t} \) is constant over time. Even though models like
(\ref{Yspt})--(\ref{gltm}) can often fit data well over a longer period of time, the
assumption of homogeneity is too restrictive in practical applications: to
guarantee a sufficient amount of data for sufficiently precise estimation, these
models are often applied over time spans of many years.
On the contrary, the strategy pursued here requires only local time
homogeneity, which means that at each time point \( t \)  there is a (possibly
rather short) interval \( [t-m,t] \), where the process \( Y_{t} \)  is well
described by model (\ref{timemodel})--(\ref{Xtparam}). This strategy aims then
both at finding an interval of homogeneity (preferably as long as possible) and
at the estimation of the corresponding parameter values \( \thetav \),
which then enable predicting \( Y_{t} \) and \( X_{t} \).


Next, we discuss the parameter estimation for model (\ref{timemodel})--(\ref{Xtparam})
using observations \( Y_{t} \) from some time interval \( I = [t_{0},t_{1}] \).
The conditional distribution of each observation \( Y_{t} \) given the past 
\( \cc{F}_{t-1} \) is determined by 
the structural variable \( X_{t} \), whose dynamics is described by the parameter vector
\( \thetav \): \( X_{t} = X_{t}(\thetav) \) for \( t \in I \) due to (\ref{Xtparam}). 
We denote the underlying value of \( \thetav \) by \( \thetavs \).

For estimating \( \thetavs \), we apply the quasi maximum likelihood (quasi-MLE) approach
using the estimating equations generated under the assumption of Gaussian errors \( \eps_{t} \). 
This guarantees efficiency under the normality of innovations and consistency under rather 
general moment conditions (Hansen and Lee, 1994; Francq and Zakoian, 2007). The 
log-likelihood for the model (\ref{timemodel})--(\ref{Xtparam}) 
on an interval \( I \)  can be represented in the form
\begin{eqnarray*} 
    L_{I}(\thetav)
    = 
    \sum_{t \in I} \ell \{ Y_{t},g[X_{t}(\thetav)] \}
\end{eqnarray*} 
with log-likelihood function 
\( \ell(y,\upsilon) = -0.5 \left\{ \log (\upsilon) + y^{2}/\upsilon \right\} \).
We define the quasi-MLE estimate \( \tilde{\thetav}_{I} \) of the parameter
\( \thetav \) by maximizing the log-likelihood \( L_{I}(\thetav) \),
\begin{eqnarray}
\label{hatthetaI}
    \tilde{\thetav}_{I} = \argmax_{\thetav \in \Theta} L_{I}(\thetav)
    = 
    \argmax_{\thetav\in \Theta} \sum_{t \in I}
    \ell\{ Y_{t},g[X_{t}(\thetav)] \},
\end{eqnarray}
and denote by \( L_{I}(\tilde{\thetav}_{I}) \) the corresponding maximum.


To characterize the quality of estimating the parameter vector
\( \thetavs = (\omega,\alpha_{1},\ldots,\alpha_{p}, \) \( \beta_{1},\ldots,\beta_{q})^{\T} \)
by \( \tilde{\thetav}_{I} \), we now present an exact (nonasymptotic) exponential risk 
bound. This bound concerns the value of maximum
\( L_{I}(\tilde{\thetav}_{I}) = \max_{\thetav \in \Theta} L_{I}(\thetav) \)
rather than the point of maximum \( \tilde{\thetav}_{I} \).
More precisely, we consider difference
\( L_{I}(\tilde{\thetav}_{I},\thetavs)
= L_{I}(\tilde{\thetav}_{I}) - L_{I}(\thetavs) \).
By definition, this value is non-negative and represents the deviation of the
maximum of the log-likelihood process from its value at the ``true''
point \( \thetavs \).
Later, we comment on how the accuracy of estimation of the parameter
\( \thetavs \) by \( \tilde{\thetav}_{I} \) relates to the value
\( L_{I}(\tilde{\thetav}_{I},\thetavs) \). We will also see that the bound for
\( L_{I}(\tilde{\thetav}_{I},\thetavs) \) yields the confidence set for the
parameter \( \thetavs \), which will be used for the proposed change-point
test. Now, the nonasymptotic risk bound is specified in the following theorem,
which formulates Corollary 4.2 and 4.3 of Spokoiny (2009) for the case of quasi-MLE 
estimation of a CH model (\ref{timemodel})--(\ref{Xtparam}) at \( \thetav = \thetavs \).
The result can be viewed as an extension of the Wilks phenomenon that the 
distribution of \( L_{I}(\tilde{\thetav}_{I},\thetavs) \) for a linear Gaussian model 
is \( \chi^{2}_{p}/2 \), where \( p \) is the number of estimated parameters in the model.
\begin{theorem}
\label{TGARCH11}
Assume that the process \( Y_{t} \) follows the model
(\ref{timemodel})--(\ref{Xtparam}) with the parameter 
\( \thetavs \in \Theta \), where the set \( \Theta \) is compact. 
The function \( g(\cdot) \) is assumed to be continuously differentiable with the 
uniformly bounded first derivative and \( g(x) \ge \delta > 0 \) for all \( x \).
Further, let the process \( X_{t}(\thetav) \) be sub-ergodic in the sense that for any
smooth function \( f(\cdot) \) there exists \( f^{*} \) such that for any time interval 
\( I \)
\begin{eqnarray*}
    \E_{\thetavs} \biggl| 
        \sum_{I} \bigl\{ f(X_{t}(\thetav)) - \E_{\thetavs} f(X_{t}(\thetav)) \bigr\} 
    \biggr|^{2}
    \le 
    f^{*} |I| ,
    \qquad 
    \thetav \in \Theta .
\end{eqnarray*}    
Let finally \( \E \exp \{ \kappa (\eps_{t}^{2}-1) | \cc{F}_{t-1} \} \le 
c(\kappa) \)  for some \( \kappa > 0 \), \( c(\kappa) > 0 \), and all \( t\in N \). 
Then there are \( \lambda > 0 \) and \( \lexpB(\lambda,\thetavs) > 0 \)  
such that for any interval \( I \) and \( \zz > 0 \) 
\begin{eqnarray} 
\label{thm21eq2}
    \P_{\thetavs}\bigl( L_{I}(\tilde{\thetav}_{I},\thetavs) > \zz \bigr)
    \le
    \exp\{ \lexpB(\lambda,\thetavs) - \lambda \zz \}.
\end{eqnarray}
Moreover, for any \( r > 0 \), there is a constant \( \Crlp_{r}(\thetavs) \) 
such that
\begin{eqnarray}
\label{THM21EQ3}
    \E_{\thetavs} \bigl| L_{I}(\tilde{\thetav}_{I},\thetavs) \bigr|^{r}
    \le
    \Crlp_{r}(\thetavs).
\end{eqnarray}
\end{theorem}

\begin{remark}
\label{RLDB}
The condition \( g(x) \ge \delta > 0 \) guarantees that the
variance process cannot reach zero. In the case of GARCH, it is sufficient to
assume \( \omega > 0 \), for instance. 
\end{remark}

One attractive feature of Theorem~\ref{TGARCH11}, formulated in the following
corollary, is that it enables constructing the non-asymptotic confidence
sets and testing the parametric hypothesis on the basis of the fitted
log-likelihood \( L_{I}(\tilde{\thetav}_{I},\thetav) \). This feature is
especially important for our procedure presented in Section~\ref{Sprocedure}.

\begin{corollary}
\label{CGARCH11}
Under the assumptions of Theorem~\ref{TGARCH11},
let the value \( \zz_{\alpha} \) fulfill
\( \lexpB(\lambda,\thetavs) - \lambda \zz_{\alpha} < \log \alpha \) for some
\( \alpha < 1 \). Then the random set
\( 
\cc{E}_{I}(\zz_{\alpha})
=
\{ \thetav: L_{I}(\tilde{\thetav}_{I},\thetav) \le \zz_{\alpha} \}
\) 
is an \( \alpha \)-confidence set for \( \thetavs \) in the sense that
\(
\P_{\thetavs}( \thetavs \not\in \cc{E}_{I}(\zz_{\alpha}) ) \le \alpha.
\)
\end{corollary}

Theorem \ref{TGARCH11} also gives a non-asymptotic and fixed upper bound
for the risk of estimation \( L_{I}(\tilde{\thetav}_{I},\thetavs) \)
that applies to an arbitrary sample size
\( |I| \). To understand the relation of this result to the classical rate
result, we can apply the standard arguments based on the quadratic expansion of the
log-likelihood \( L(\tilde{\thetav},\thetav) \).
Let \( \nabla^{2}L(\thetav) \) denote the Hessian matrix of the second
derivatives of \( L(\thetav) \) with respect to the parameter \( \thetav \).
Then
\begin{eqnarray}
\label{Lquadro}
    L_{I}(\tilde{\thetav}_{I},\thetavs)
    =
    0.5
    \bigl( \tilde{\thetav}_{I} - \thetavs \bigr)^{\T}
    \nabla^{2}L_I(\thetav'_{I})
    \bigl( \tilde{\thetav}_{I} - \thetavs \bigr),
\end{eqnarray}
where \( \thetav'_{I} \) is a convex combination of \( \thetavs \) and \( \tilde{\thetav}_{I} \). 
Under usual regularity assumptions and for sufficiently large \( |I| \),
the normalized matrix \( |I|^{-1} \nabla^{2} L_I(\thetav) \) is close to some matrix
\( V(\thetav) \), which depends only on the stationary distribution of
\( Y_{t} \) and is continuous in \( \thetav \). Then (\ref{thm21eq2}) 
approximately means that
\( 
    \bigl\| \sqrt{V(\thetavs)} ( \tilde{\thetav}_{I} - \thetavs ) \bigr\|^{2}
    \le
    \zz/|I|
\) 
with probability close to 1 for large \( \zz \). Hence, the large
deviation result of Theorem~\ref{TGARCH11} yields the root-\( |I| \)
consistency of the MLE estimate \( \tilde{\thetav}_{I} \). See Spokoiny (2009)
for further details.

\section{Pointwise adaptive nonparametric estimation}
\label{Sprocedure}

An obvious feature of the model (\ref{timemodel})--(\ref{Xtparam}) is that the
parametric structure of the process is assumed constant over the whole sample
and cannot thus incorporate changes and structural breaks at unknown times in
the model. A natural generalization leads to models whose coefficients may change
over time (Fan and Zhang, 2008).
%
One can then assume that the structural process \( X_{t} \) satisfies the
relation (\ref{Xtparam}) at any time, but the vector of coefficients \( \thetav \)
may vary with the time \( t \), \( \thetav = \thetav(t) \). The estimation of
the coefficients as general functions of time is possible only under some
additional assumptions on these functions. Typical assumptions are (i) varying
coefficients are smooth functions of time (Cai et al., 2000) and (ii) varying
coefficients are piecewise constant functions (Bai and Perron, 1998; Mikosch
and Starica, 1999, 2004).

Our {local parametric approach} differs from the commonly used identification
assumptions (i) and (ii). We assume that the observed data \( Y_{t} \) are
described by a (partially) unobserved process \( X_{t} \) due to
(\ref{timemodel}), and at each point \( T \), there exists a historical
interval \( I(\TT) = [t_{0},\TT] \) in which the process \( X_{t} \) ``nearly''
follows the parametric specification (\ref{Xtparam}) (see Section
\ref{Sresults} for details on what ``nearly'' means). This local structural
assumption enables us to apply well developed parametric estimation for data
\( \{ Y_{t} \}_{t \in I(\TT)} \) to estimate the underlying parameter \( \thetav =
\thetav(\TT) \) by \( \hat{\thetav} = \hat{\thetav}(\TT) \). (The estimate 
\( \hat{\thetav} = \hat{\thetav}(\TT) \) can be then used for estimating the value
\( \hat{X}_{\TT} \) of the process \( X_{t} \) at \( \TT \) from equation
(\ref{Xtparam}) and for further modeling such as forecasting \( Y_{\TT+1} \)).
Moreover, this assumption includes the above mentioned ``smooth transition''
and ``switching regime'' assumptions (i) and (ii) as special cases: parameters
\( \hat{\thetav}(\TT) \) vary over time as the interval \( I(\TT) \) changes
with \( \TT \), and at the same time, discontinuities and jumps in \(
\hat{\thetav}(\TT) \) as a function of time are possible.

To estimate \( \hat{\thetav}({\TT}) \), we have to find the historical interval of
homogeneity \( I(\TT) \), that is, the longest interval \( I \) with the right-end
point \( \TT \), where data do not contradict a specified parametric model with
fixed parameter values. Starting at each time \( \TT \) with a very short interval
\( I = [t_{0},\TT] \), we search by successive extending and testing of interval \( I \) on
homogeneity against a change-point alternative: if the hypothesis of
homogeneity is not rejected for a given \( I \), a larger interval is taken and
tested again. Contrary to Bai and Perron (1998) and Mikosch and Starica (1999),
who detect all change points in a given time series, our approach is local: it
focuses on the local change-point analysis near the point \( \TT \) of estimation
and tries to find only one change closest to the reference point.

In the rest of this section, we first discuss the test statistics employed to
test the time-homogeneity of an interval \( I \) against a change-point alternative
in Section~\ref{Stest}. Later, we rigorously describe the pointwise adaptive 
estimation procedure in Section~\ref{Ssprocedure}. Its implementation and the
choice of parameters entering the adaptive procedure are described in
Sections \ref{Ssprocedure}--\ref{Ssrrho}. Theoretical properties of the method
are studied in Section~\ref{Sresults}.

\subsection{Test of homogeneity against a change-point alternative}
\label{Stest}

The pointwise adaptive estimation procedure crucially relies on the test of
local time-homogeneity of an interval \( I = [t_{0},\TT] \). The null hypothesis
for \( I \) means that the observations \( \{ Y_{t} \}_{t\in I} \) follow the
parametric model (\ref{timemodel})--(\ref{Xtparam}) with a fixed parameter \( 
\thetavs \), leading to the quasi-MLE estimate \( \tilde{\thetav}_{I} \) from (\ref{hatthetaI}) 
and the corresponding fitted log-likelihood \( L_{I}(\tilde{\thetav}_{I}) \).

The change-point alternative for a given change-point location \( \tcp \in I \)
can be described as follows:
process \( Y_{t} \) follows the parametric model (\ref{timemodel})--(\ref{Xtparam})
with a parameter \( \thetav_{J} \) for \( t \in J = [t_{0},\tcp] \) and with a
different parameter \( \thetav_{J^{c}} \) for \( t \in J^{c} = [\tcp+1,\TT] \);
\( \thetav_{J} \ne \thetav_{J^{c}} \). 
The fitted log-likelihood under this alternative reads as \( L_{J}(\tilde{\thetav}_{J}) +
L_{J^{c}}(\tilde{\thetav}_{J^{c}}) \). The test of homogeneity can be performed
using the likelihood ratio (LR) test statistic 
\( T_{I,\tcp} \):
\begin{eqnarray*}
    T_{I,\tcp}
    =
    \max_{\thetav_{J}, \thetav_{J^{c}} \in \Theta}
    \left\{ L_{J}(\thetav_{J}) + L_{J^{c}}(\thetav_{J^c}) \right\}
    - \max_{\thetav \in \Theta} L_{I}(\thetav)
    =
    \bigl\{ 
        L_{J}(\tilde{\thetav}_{J}) + {L}_{J^{c}}(\tilde{\thetav}_{J^c}) 
        - {L}_{I}(\tilde{\thetav}_{I}) 
    \bigr\}.
\end{eqnarray*}
Since the change-point location \( \tcp \) is generally not known, 
we consider the supremum
of the LR statistics \( T_{I,\tcp} \) over some subset \( \tcp\in\Tcp(I) \),
cf. Andrews (1993):
\begin{eqnarray}
\label{LRtest}
    T_{I,\Tcp(I)}
    =
    \sup_{\tcp \in \Tcp(I)} T_{I,\tcp} \, .
\end{eqnarray}    
A typical example of a set \( \Tcp(I) \) is
\( \Tcp(I) = \{ \tau : t_{0} + m' \le \tcp \le \TT - m'' \ \} \) for some fixed 
\( m', m'' > 0 \).

\subsection{Adaptive search for the longest interval of homogeneity}
\label{Ssprocedure}

This section presents the proposed adaptive pointwise estimation procedure. 
At each point \( \TT \),
we aim at estimating the unknown parameters \( \thetav(\TT) \) from
historical data \( Y_{t}, \,  t \le \TT \); this procedure repeats for every current
time point \( \TT \) as new data arrives. At the first step, the procedure selects
on the base of historical data an interval \( \hat{I}(\TT) \) of homogeneity 
in which the data do not contradict the parametric model
(\ref{timemodel})--(\ref{Xtparam}). Afterwards, the quasi-MLE estimation is
applied using the selected historical interval \( \hat{I}(\TT) \) to obtain estimate 
\( \hat{\thetav}(\TT) = \tilde{\thetav}_{\hat{I}(\TT)} \). From now on, we consider
an arbitrary, but fixed time point \( \TT \).

Suppose that a growing set \( I_{0} \subset I_{1} \subset \ldots \subset I_{\K} \)
of historical interval-candidates \( I_{k} = [\TT-m_{k}+1,\TT] \) with 
the right-end point \( \TT \) is fixed. 
The smallest interval \( I_{0} \) is accepted automatically as homogeneous. 
Then the procedure successively checks every larger interval \( I_{k} \) on 
homogeneity using the test statistic \( T_{I_{k},\Tcp(I_{k})} \) from (\ref{LRtest}).
The selected interval \( \hat{I} \) corresponds to the largest accepted interval 
\( I_{\hat{k}} \) with index \( \hat{k} \) such that
\begin{eqnarray}\label{eqcritvalintro}
    T_{I_{k},\Tcp(I_{k})}
    \le 
    \zz_{k},
    \qquad 
    k \le \hat{k} ,
\end{eqnarray}    
and \( T_{I_{\hat{k}+1},\Tcp(I_{\hat{k}+1})} > \zz_{\hat{k}+1} \), where 
the critical values \( \zz_{k} \) are discussed
later in this section and specified in Section~\ref{Slambda}. 
This procedure then leads to the adaptive estimate 
\( \hat{\thetav} = \tilde{\thetav}_{\hat{I}} \) corresponding to the selected 
interval \( \hat{I} = I_{\hat{k}} \).

The complete description of the procedure includes two steps. 
(A) Fixing the set-up and the parameters of the procedure.
(B) Data-driven search for the longest interval of homogeneity.

\begin{description}
    \item[(A) Set-up and parameters:] 
        \begin{enumerate}
            \item Select a specific parametric model (\ref{timemodel})--(\ref{Xtparam}) 
            (e.g., constant volatility, ARCH(1), GARCH(1,1)).
            \item Select the set \( \IK = (I_{0},\ldots,I_{\K}) \) of interval-candidates, 
            and for each \( I_{k} \in \IK \), the set \( \Tcp(I_{k}) \) 
            of possible change points \( \tcp \in I_{k} \) 
            used in the LR test (\ref{LRtest}). 
            \item Select 
            the critical values \( \zz_{1},\ldots,\zz_{\K} \) in (\ref{eqcritvalintro}) 
            as described in Section~\ref{Slambda}.
        \end{enumerate}
    \item[(B) Adaptive search and estimation:] 
        Set \( k=1 \), \( \hat{I} = I_{0} \),  and 
        \( \hat{\thetav} = \tilde{\thetav}_{I_{0}} \). 
        \begin{enumerate}
            \item Test the hypothesis \( H_{0,k} \) of no change point within the 
            interval \( I_{k} \) using test statistics (\ref{LRtest}) and the critical values 
            \( \zz_{k} \) obtained in (A3). 
            If a change point is detected 
            (\( H_{0,k} \) is rejected), go to (B3). 
            Otherwise proceed with (B2).             
            
            \item 
            Set 
            \( \hat{\thetav} = \tilde{\thetav}_{I_{k}} \) and 
            \( \hat{\thetav}_{I_{k}} = \tilde{\thetav}_{I_{k}} \).
            Further, set \( k := k + 1 \). If \( k \le \K \), repeat (B1);
            otherwise go to (B3).
            
            \item Define \( \hat{I} = I_{k-1} = \) ``the last accepted 
            interval'' 
            and \( \hat{\thetav} = \tilde{\thetav}_{\hat{I}} \). 
            Additionally, set \( \hat{\thetav}_{I_{k}} = \ldots = \hat{\thetav}_{I_{\K}} = \hat{\thetav} 
            \) if \( k\le \K\).
            
        \end{enumerate}
\end{description}

In the step (A), one has to select three main ingredients of the
procedure. First, the parametric model used locally to approximate the process
\( Y_{t} \) has to be specified in (A1), for example, the constant volatility or
GARCH(1,1) in our context. Next in step (A2), the set of intervals 
\( \IK = \{I_{k}\}_{k=0}^{\K} \) 
is fixed, each interval with the right-end point \( \TT \), length \( m_{k} = |I_{k}| \),
and the set \( \Tcp(I_{k}) \) of tested change points.
Our default proposal is to use a geometric grid \( m_{k} = [m_{0} a^{k}], a>1, \) and to set 
\( I_{k} = [\TT-m_{k}+1,\TT] \) and \( \Tcp(I_{k}) = [\TT-m_{k-1}+1,\TT-m_{k-2}] \).
Although our experiments show that the procedure is rather
insensitive to the choice of \( m_{0} \) and \( a \) (e.g., we use \( m_{0}=10 \) and \( a=1.25 \)
in simulations), the length \( m_{0} \) of interval \( I_{0} \) should take into account
the parametric model selected in (A1). The reason is that \( I_{0} \) is always
assumed to be time-homogeneous and \( m_{0} \) thus has to reflect flexibility of the
parametric model; for example, while \( m_{0} = 20 \) might be reasonable for
GARCH(1,1) model, \( m_{0} = 5 \) could be a reasonable choice for the locally
constant approximation of a volatility process. Finally in step (A3), one has to select
the \( \K \) critical values \( \zz_{k} \) in (\ref{eqcritvalintro})
for the LR test statistics \( T_{I_{k},\Tcp(I_{k})} \) from (\ref{LRtest}).
The critical values \( \zz_{k} \) will generally depend on the parametric model describing 
the null hypothesis of time-homogeneity, the set \( \IK \) of intervals \( I_{k} \)
and corresponding sets of considered change points \( \Tcp(I_{k}) \), \( k \le \K \), and 
additionally, on two constants \( r \) and \( \alpn \) that are counterparts of the usual 
significance level. 
All these determinants of the critical values can be selected in step (A) and
the critical values are thus obtained before the actual estimation takes place in step (B).
Due to its importance, the method of constructing critical values 
\( \{ \zz_{k} \}_{k=1}^{\K} \) 
is discussed separately in Section~\ref{Slambda}.

The main step (B) performs the search for the longest time-homogeneous interval.
Initially, \( I_{0} \) is assumed to be homogeneous. 
If \( I_{k-1} \) is negatively tested on the presence of a change point, one 
continues with \( I_{k} \) by employing the test (\ref{LRtest}) in step (B1),
which checks for a potential change point in \( I_{k} \).
If no change point is found, then \( I_{k} \) is accepted as time-homogeneous in step (B2);
otherwise the procedure terminates in step (B3). 
We sequentially repeat these tests until 
we find a change point or exhaust all intervals. The latest (longest) interval
accepted as time-homogeneous is used for estimation in step (B3). Note that the
estimate \( \hat{\thetav}_{I_{k}} \) defined in (B2) and (B3) corresponds to
the latest accepted interval \( \hat{I}_{k} \) after the first \( k \) steps, 
or equivalently, the interval selected out of \( I_{1},\ldots,I_{k} \). 

Moreover, the whole search and estimation step (B) can be repeated at different time points 
\( \TT \) without reiterating the initial step (A) as the critical values \( \zz_{k} \) 
depend only on the approximating parametric model and 
interval lengths \( m_{k} = |I_{k}| \), not on the time point \( \TT \) (see Section~\ref{Slambda}).

\subsection{Choice of critical values \( \zz_{k} \) }
\label{Slambda}

The presented method of choosing the interval of homogeneity \( \hat{I} \) can 
be viewed as multiple testing procedure. The critical values for this procedure 
are selected using the general approach of testing theory: to provide a 
prescribed performance of the procedure under the null hypothesis, that is, in 
the pure parametric situation. 
This means that the procedure is trained on the data generated from 
the pure parametric time homogeneous model from step (A1). The correct choice in this 
situation is the largest considered interval \( I_{\K} \) and a choice \( I_{\hat{k}} \)
with \( \hat{k} < \K \) can be interpreted as  a ``false alarm''.
We select the minimal critical values 
ensuring a small probability of such a false alarm.
%
Our condition slightly differs though from the classical level condition 
because we focus on parameter estimation rather than on hypothesis testing. 

In the pure parametric case, the ``ideal'' estimate corresponds to the largest 
considered interval \( I_{\K} \). 
Due to Theorem~\ref{TGARCH11}, the quality of estimation of the parameter \( \thetavs \) 
by \( \tilde{\thetav}_{I_{\K}} \) can be measured by the log-likelihood ``loss'' 
\( L_{I_{\K}}(\tilde{\thetav}_{I_{\K}},\thetavs) \), which is stochastically bounded with 
exponential and polynomial moments: 
\( \E_{\thetavs} | L_{I_{\K}}( \tilde{\thetav}_{I_{\K}}, \thetavs ) |^{r} \le 
\Crlp_{r}(\thetavs) \). 
If the adaptive procedure stops earlier at some intermediate step \( k < \K \), we select 
instead of \( \tilde{\thetav}_{I_{\K}} \) another estimate \( \hat{\thetav} = \tilde{\thetav}_{I_{k}} \) 
with a larger variability. The loss associated with such a false alarm can be measured by 
the value \( L_{I_{\K}}(\tilde{\thetav}_{I_{\K}},\hat{\thetav}) 
= L_{I_{\K}}(\tilde{\thetav}_{I_{\K}}) - L_{I_{\K}}(\hat{\thetav}) \). 
The corresponding condition bounding the loss due to the adaptive estimation reads as
\begin{eqnarray}
\label{thm21eq3}
    \E_{\thetavs} \bigl| L_{I_{\K}}(\tilde{\thetav}_{I_{\K}},\hat{\thetav}) \bigr|^{r}
    \le
    \alpn \Crlp_{r}(\thetavs) .
\end{eqnarray}    
This is in fact an implicit condition on the critical values \( \{\zz_{k}\}_{k=1}^{\K} \), 
which ensures that the loss associated with the false alarm is at most the 
\( \alpn \)-fraction of the log-likelihood loss of the ``ideal'' or ``oracle'' 
estimate \( \tilde{\thetav}_{I_{\K}} \) for the parametric situation. 
The constant \( r \) corresponds to the power of the loss in (\ref{thm21eq3}),
while \( \alpn \) is similar in meaning to the test level. In the limit 
case when \( r \) tends to zero, this condition (\ref{thm21eq3}) becomes the usual level 
condition: 
\( \P_{\thetavs}(I_{\K} \text{ is rejected}) 
= \P_{\thetavs} \bigl( \tilde{\thetav}_{I_{\K}} \ne \hat{\thetav} \bigr) 
\le \alpn \). 
The choice of the metaparameters \( r \) and \( \alpn \) is discussed in 
Section~\ref{Ssrrho}.

A condition similar to (\ref{thm21eq3}) is imposed at each step of the adaptive procedure.
The estimate \( \hat{\thetav}_{I_{k}} \) coming after the \( k \) steps of the procedure
should satisfy
\begin{eqnarray}
\label{defcv}
    \E_{\thetavs} \bigl| L_{I_{k}}(\tilde{\thetav}_{I_{k}},\hat{\thetav}_{I_{k}}) \bigr|^{r}
    \le
    \alpn_{k} \Crlp_{r}(\thetavs) ,
    \qquad 
    k=1,\ldots,\K,
\end{eqnarray}    
where \( \alpn_{k} = \alpn k / \K \le \alpn \). 
The following theorem presents some sufficient conditions on the critical
values \( \{\zz_{k}\}_{k=1}^{\K} \) ensuring (\ref{defcv}); recall that
\( m_{k} = |I_{k}| \) denotes the length of \( I_{k} \).

\begin{theorem}
\label{TCVGARCH}
Suppose that \( r > 0 \), \( \alpn > 0 \).
Under the assumptions of Theorem \ref{TGARCH11}, there are constants
\( a_{0}, a_{1}, a_{2} \) such that the condition (\ref{defcv}) is fulfilled 
with the choice 
\begin{eqnarray*}
    \zz_{k}
    =
    a_{0} r \log(\alpn^{-1}) 
    + a_{1} r \log (m_{\K}/m_{k-1})
    + a_{2} \log(m_{k}),
    \qquad 
    k = 1,\ldots,\K.
\end{eqnarray*}    
\end{theorem}

Since \( \K \) and \( \{m_{k}\}_{k=1}^{\K} \) are fixed, 
the \( \zz_{k} \)'s in Theorem~\ref{TCVGARCH} have a form 
\( \zz_{k} = C + D \log(m_{k}) \) for \( k=1,\ldots,\K \) 
with some constant \( C \) and \( D \). 
However, a practically relevant choice of these constants has to be done by
Monte-Carlo simulations.
Note first that every particular choice of the coefficients \( C \) and \( D \) determines the 
whole set of the critical values \( \{\zz_{k}\}_{k=1}^{\K} \) and thus the local change-point procedure. 
For the critical values given by fixed \( (C,D) \), one can run the procedure and observe 
its performance on the simulated data using the data-generating process 
(\ref{timemodel})--(\ref{Xtparam}); in particular, one can check whether the
condition (\ref{defcv}) is fulfilled. For any (sufficiently large) fixed value of \( C \), 
one can thus find the minimal value \( D(C) < 0 \) of \( D \) that ensures (\ref{defcv}). 
Every corresponding set of critical values in the form \( \zz_{k} = C + D(C) \log(m_{k}) \) is 
admissible. The condition \( D(C) < 0 \) ensures that the critical values decreases 
with \( k \). This reflects the fact that a false alarm at an early stage of the 
algorithm is more crucial because it leads to the choice of a highly variable 
estimate. The critical values \( \zz_{k} \) for small \( k \) should thus be rather 
conservative to provide the stability of the algorithm in the parametric situation.
To determine \( C \), the value \( \zz_{1} \) can be fixed by considering the false alarm 
at the first step of the procedure, which leads to estimation using the smallest  
interval \( I_{0} \) instead of the ``ideal'' largest interval \( I_{\K} \).
The related condition (used in Section \ref{Sfincv}) reads as 
\begin{eqnarray}\label{eqfixcv}
    \E_{\thetavs} \bigl| 
        L_{I_{\K}}(\tilde{\thetav}_{I_{\K}},\tilde{\thetav}_{I_{0}}) 
    \bigr|^{r} \bb{1}(T_{I_{1},\Tcp(I_{1})} > \zz_{1})
    \le
    \alpn \Crlp_{r}(\thetavs) / \K .
\end{eqnarray}    
Alternatively, one could select a pair \( (C,D) \) that minimizes the resulting
prediction error, see Section~\ref{Ssrrho}.

\subsection{Selecting parameters \( r \) and \( \alpn \)}
\label{Ssrrho}

The choice of critical values using inequality (\ref{defcv}) additionally
depends on two ``metaparameters'' \( r \)  and \( \alpn \). A simple strategy
is to use conservative values for these parameters and the corresponding set
of critical values (e.g., our default is \( r=1 \) and \( \alpn=1 \)). 
On the other hand, the two parameters are global in the
sense that they are independent of \( \TT \). Hence, one can also determine
them in a data-driven way by minimizing some global forecasting error
(Cheng et al., 2003). Different values of \( r \) and \( \alpn \) may lead to
different sets of critical values and hence to
different estimates \( \hat{\thetav}^{(r,\alpn)}(\TT) \) and to different
forecasts \( \hat{Y}_{\TT+h|\TT}^{(r,\alpn)} \) of the future values \( Y_{\TT+h} \),
where \( h \) is the forecasting horizon.
Now, a data-driven choice of \( r \) and \( \alpn \) can be done by
minimizing the following objective function:
\begin{eqnarray}
\label{eqforerr}
    (\hat{r},\hat{\alpn})
    =
    \argmin_{r>0,\alpn>0} PE_{\Lambda,\cc{H}}(r,\alpn)
    =
    \argmin_{r,\alpn} \sum_{\TT} \sum_{h \in \cc{H}}
    \Lambda\bigl(  Y_{\TT+h}, \hat{Y}_{\TT+h|\TT}^{(r,\alpn)}  \bigr),
\end{eqnarray}
where \( \Lambda \)  is a loss function and \( \cc{H} \) is the
forecasting horizon set. For example, one can take
\( \Lambda_r(\upsilon,\upsilon')=|\upsilon-\upsilon'|^{r} \)
for \( r \in [1/2,2] \). 
For daily data, the forecasting horizon could be one day, \( \cc{H}=\{1\} \),
or two weeks, \( \cc{H}=\{1,\ldots,10\} \).

\section{Theoretic properties}
\label{Sresults}

In this section, we collect basic results describing the quality of the proposed
adaptive procedure. First, the definition of the procedure ensures the
performance prescribed by (\ref{defcv}) in the parametric situation. We however
claimed that the adaptive pointwise estimation applies even if the process \(
Y_{t} \) is only locally approximated by a parametric model. Therefore, we now
define locally ``nearly parametric'' process, for which we derive an analogy of
Theorem \ref{TGARCH11} (Section \ref{secsmb}). Later, we prove certain
``oracle'' properties of the proposed method (Section \ref{secoracle}).

\subsection{Small modeling bias condition}\label{secsmb}

This section discusses the concept of ``nearly parametric'' case. To define it
rigorously, we have to quantify the quality of approximating the true latent process \(
X_{t} \), which drives the observed data \( Y_{t} \) due to (\ref{timemodel}), 
by the parametric process \( X_{t}(\thetav) \) described by (\ref{Xtparam}) for
some \(\thetav \in \Theta\). 
Below we assume that the innovations \( \varepsilon_{t} \) in the model (\ref{timemodel}) 
are independent and identically distributed and denote the distribution of 
\( \sqrt{\upsilon} \varepsilon_{t} \) by \( P_{\upsilon} \) so that the conditional 
distribution of \( Y_{t} \) given \( \cc{F}_{t-1} \) is \( P_{g(X_{t})} \). 
%
To measure the distance of a data-generating process from a parametric model, we introduce 
for every interval \( I_{k} \in \IK \) and every parameter \( \thetav \in \Theta \) the random quantity
\begin{eqnarray*}
    \Delta_{I_{k}}(\thetav)
    =
    \sum_{t \in I_{k}} \kullb\{ g(X_{t}), g[X_{t}(\thetav)] \},
\end{eqnarray*}
where \( \kullb(\upsilon,\upsilon') \) denotes the Kullback-Leibler distance between 
\( P_{\upsilon} \)  and \( P_{\upsilon'} \).
For CH models with Gaussian innovations \( \eps_{t} \),
\( \kullb(\upsilon,\upsilon') = -0.5 \{ \log(\upsilon/\upsilon') + 1 - \upsilon/\upsilon' \} \).
In the parametric case with \( X_{t} = X_{t}(\thetavs) \), we clearly have
\( \Delta_{I_{k}}(\thetavs) = 0 \). To characterize the ``nearly parametric case,''
we introduce {small modeling bias} (SMB) condition, which simply means
that, for some \( \thetav \in \Theta \),
\( \Delta_{I_{k}}(\thetav) \) is bounded by a small constant with a high probability.
Informally, this means that the ``true'' model can be well approximated on the interval
\( I_{k} \) by the parametric one with the parameter \( \thetav \). 
The best parametric fit (\ref{Xtparam}) to the underlying model (\ref{timemodel}) 
on \( I_{k} \) can be defined by 
minimizing the value \( \E \Delta_{I_{k}}(\thetav) \) over \( \thetav \in \Theta \)
and \( \tilde{\thetav}_{I_{k}} \) can be viewed as its estimate.

The following theorem claims that the results on the accuracy of estimation
given in Theorem~\ref{TGARCH11} can be extended from the parametric case to the
general nonparametric situation under the SMB condition.
Let \( \losst(\hat{\thetav},\thetav) \) be any loss function
for an estimate \( \hat{\thetav} \).

\begin{theorem}
\label{Triskbound}
Let for some \( \thetav \in \Theta \) and some  \( \Delta \ge 0 \)
\begin{eqnarray}
\label{DeltaI}
\E \Delta_{I_{k}}(\thetav)
\le \Delta.
\end{eqnarray}
Then it holds for an estimate \( \hat{\thetav} \) constructed from the observations
\( \{ Y_{t} \}_{t \in I_{k}} \) that
\begin{eqnarray*}
    \E \log \bigl(
        1 + \losst(\hat{\thetav},\thetav) / \E_{\thetav} \losst(\hat{\thetav},\thetav)
    \bigr)
    \le
    1 + \Delta.
\end{eqnarray*}
\end{theorem}

This general result applied to the quasi-MLE estimation with the loss
function \( L_{I} (\tilde{\thetav}_{I} , \thetav) \)
yields the following corollary.

\begin{corollary}
\label{CSMBGARCH}
Let the SMB condition (\ref{DeltaI}) hold for some interval \( I_{k} \) and
\( \thetav \in \Theta \). Then 
\begin{eqnarray*}
    \E \log \Bigl(
        1 + \bigl| L_{I_k} (\tilde{\thetav}_{I_k} , \thetav) \bigr|^{r} / \Crlp_{r}(\thetav)
    \Bigr)
    \le
    1 + \Delta ,
\end{eqnarray*}
where \( \Crlp_{r}(\thetav) \) is the parametric risk bound from (\ref{THM21EQ3}).
\end{corollary}

This result shows that the estimation loss
\( | L_{I} (\tilde{\thetav}_{I} , \thetav) |^{r} \) normalized by
the parametric risk \( \Crlp_{r}(\thetav) \) is stochastically bounded by a
constant proportional to \( e^{\Delta} \). If \( \Delta \) is not large, this
result extends the parametric risk bound (Theorem~\ref{TGARCH11}) to the
nonparametric situation under the SMB condition. Another implication of
Corollary~\ref{CSMBGARCH} is that the confidence set built for the parametric model
(Corollary~\ref{CGARCH11}) continues to hold, with a slightly smaller coverage
probability, under SMB.

\subsection{The ``oracle'' choice and the ``oracle'' result}\label{secoracle}

Corollary~\ref{CSMBGARCH} suggests that the ``optimal'' or ``oracle''
choice of the interval \( I_{k} \) from the set \( I_{1},\ldots,I_{\K} \) can be
defined as the largest interval for which the SMB condition (\ref{DeltaI}) still
holds (for a given small \( \Delta>0 \)). For such an interval, one can neglect
deviations of the underlying process 
from a parametric model with a fixed parameter~\( \thetav \).
Therefore, we say that the choice \( \ko \) is the ``oracle'' choice if there exists 
\( \thetav \in \Theta \) such that
\begin{eqnarray}
\label{koGARCH}
    \E \Delta_{I_{\ko}}(\thetav)
    \le \Delta,
\end{eqnarray}
for a fixed \( \Delta > 0 \) and that (\ref{koGARCH}) does not hold for \( k > \ko \).
Unfortunately, the underlying process \( X_{t} \) and hence, the value 
\( \Delta_{I_{k}} \) is unknown and the oracle choice cannot be implemented.
The proposed adaptive procedure tries to mimic this oracle on the basis of available data using 
the sequential test of homogeneity. The final oracle result claims that the adaptive 
estimate  provides the same (in order) accuracy as the oracle one.

By construction, the pointwise adaptive procedure described in Section~\ref{Sprocedure}
provides the prescribed performance if the underlying process follows the parametric
model (\ref{Yproc}). Now, condition (\ref{defcv}) combined with 
Theorem~\ref{Triskbound} implies similar performance in the first \( \ko \) steps 
of the adaptive estimation procedure.

\begin{theorem}
\label{TpropGARCH}
Let \(  \thetav \in \Theta  \) and \( \Delta > 0 \) be such that 
\( \E \Delta_{I_{\ko}}(\thetav) \le \Delta \) for some \( \ko \le \K \). 
Also let 
\( \max_{k \le \ko} \E_{\thetav} | L_{I_{k}}( \tilde{\thetav}_{I_{k}}, \thetav ) |^{r} 
\le \Crlp_{r}(\thetav) \). Then
\begin{eqnarray*} 
    \E \log\biggl(
        1 +
        \frac{\bigl| L_{I_{\ko}}\bigl( \tilde{\thetav}_{I_{\ko}}, \thetav \bigr) \bigr|^{r}}{\Crlp_{r}(\thetav)}
    \biggr)
     \le 
    1 + \Delta
    \mbox{~~~and~~~} 
    \E \log\biggl(
        1 +
        \frac{\bigl| L_{I_{\ko}}\bigl( \tilde{\thetav}_{I_{\ko}}, \hat{\thetav}_{I_{\ko}} \bigr) \bigr|^{r}}
        {\Crlp_{r}(\thetav)}
    \biggr)
    \le 
    \alpn + \Delta.
\end{eqnarray*} 
\end{theorem}
Similarly to the parametric case, 
under the SMB condition \( \E \Delta_{I_{\ko}}(\thetav) \le \Delta \), any choice 
\( \hat{k} < \ko \) can be viewed as a false alarm.
Theorem \ref{TpropGARCH} documents that the loss induced by such a false alarm at the first 
\( \ko \) steps and measured by
\( L_{I_{\ko}}( \tilde{\thetav}_{I_{\ko}}, \hat{\thetav}_{I_{\ko}} ) \)
is of the same magnitude as the loss 
\( L_{I_{\ko}}( \tilde{\thetav}_{I_{\ko}}, \thetav)  \)
of estimating the parameter \( \thetav \) from the SMB (\ref{koGARCH}) by 
\( \tilde{\thetav}_{I_{\ko}} \). Thus under (\ref{koGARCH}),
the adaptive estimation during steps \( k \le\ko \) does not induce
larger errors into estimation than the quasi-MLE estimation itself.

For further steps of the algorithm with  \( k > \ko \), where (\ref{koGARCH})
does not hold, the value \( \Delta' = \E \Delta_{I_{k}}(\thetav) \) can be large and the
bound for the risk becomes meaningless due to the factor \( e^{\Delta'} \). To
establish the result about the quality of the final estimate, we thus have to
show that the quality of estimation cannot be 
destroyed at the steps \( k > \ko \).
The next ``oracle'' result states the final quality of our adaptive estimate 
\( \hat{\thetav} \).


\begin{theorem}
\label{TstabGARCH}
Let \( \E \Delta_{I_{\ko}}(\thetav) \le \Delta \) for some \( \ko \le K \).
Then \( L_{I_{\ko}}(\tilde{\thetav}_{I_{\ko}}, \hat{\thetav}) \bb{1}(\hat{k} \ge \ko) \le 
\zz_{\ko} \) yielding
\begin{eqnarray*}
    \E \log\biggl(
        1 +
        \frac{\bigl| L_{I_{\ko}}\bigl( \tilde{\thetav}_{I_{\ko}}, \hat{\thetav} \bigr) \bigr|^{r}}
        {\Crlp_{r}(\thetav)}
    \biggr)
    \le 
    \alpn + \Delta + \log\biggl( 1 + \frac{\zz_{\ko}^{r}}{\Crlp_{r}(\thetav)} \biggr).
\end{eqnarray*}    
\end{theorem}

Due to this result, the value 
\( L_{I_{\ko}}\bigl( \tilde{\thetav}_{I_{\ko}}, \hat{\thetav} \bigr) \)
is stochastically bounded. 
This can be interpreted as oracle property of \( \hat{\thetav} \) 
because it means that the adaptive estimate 
\( \hat{\thetav} \) belongs with a high probability to the confidence set 
of the oracle estimate \( \tilde{\thetav}_{I_{\ko}} \).


\section{Simulation study}
\label{Sappl}

In the last two sections, we present simulation study (Section \ref{Sappl}) and
real data applications (Section \ref{Srealappls}) documenting the performance
of the proposed adaptive estimation procedure. To verify the practical
applicability of the method in a complex setting, we concentrate on the
volatility estimation using parametric and adaptive pointwise estimation of
constant volatility, ARCH(1), and GARCH(1,1) models (for the sake of brevity,
referred to as the local constant, local ARCH, and local GARCH). The reason is
that the estimation of GARCH models requires generally hundreds of observations
for reasonable quality of estimation, which puts the adaptive procedure working
with samples as small as 10 or 20 observations to a hard test. Additionally,
the critical values obtained as described in Section \ref{Slambda} depend on
the underlying parameter values in the case of (G)ARCH.

Here we first study the finite-sample critical values for the test of
homogeneity by means of Monte Carlo simulations and discuss practical
implementation details (Section~\ref{Sfincv}). Later, we demonstrate the
performance of the proposed adaptive pointwise estimation procedure in
simulated samples (Sections~\ref{Sbreaksim}). Note that, throughout this
section, we identify the GARCH(1,1) models by triplets \( (\omega, \alpha,
\beta) \): for example, \( (1, 0.1, 0.3) \)-model. Constant volatility and
ARCH(1), are then indicated by \( \alpha = \beta = 0 \)  and \( \beta =0 \),
respectively. The GARCH estimation is done using GARCH 3.0 package (Laurent and 
Peters, 2006) and Ox 3.30 (Doornik, 2002). 
Finally, since the focus is on modelling the volatility $\sigma_{t}^{2}$
in (\ref{timemodel}), the performance measurement and comparison of all models 
at time \(t\) is done by the absolute prediction error (PE) of the volatility 
process over a prediction horizon \( \cc{H} \): 
$\mathrm{APE}(t) = \sum_{h \in \cc{H}} |\sigma_{t+h}^2 - \hat\sigma_{t+h|t}^2|/|\cc{H}|$,
where $\hat\sigma_{t+h|t}^2$ represents the volatility prediction by a particular
model.
%


\subsection{Finite-sample critical values for test of homogeneity}
\label{Sfincv}

A practical application of the pointwise adaptive procedure requires critical
values for the test of local homogeneity of a time series. Since they
are obtained under the null hypothesis that a chosen parametric model
(locally) describes the data, see Section \ref{Sprocedure}, we need to obtain
the critical values for the constant volatility, ARCH(1), and GARCH(1,1) models.
Furthermore for given \( r \) and \( \rho \), the average risk (\ref{defcv}) 
between the adaptive and oracle estimates can be bounded for critical values 
that linearly depend on the logarithm of interval length \( |I_{k}| \):
\( \zz(|I_{k}|)  = \zz_{k} = C + D \log(|I_{k}|)\)
(see Theorem~\ref{TCVGARCH}). As described in Section \ref{Slambda}, we choose 
here the smallest \( C \) satisfying (\ref{eqfixcv}) and the corresponding minimum
admissible value \( D = D(C) < 0 \) that guarantees the conditions (\ref{defcv}).

We simulated the critical values for ARCH(1) and GARCH(1,1) models with
different values of underlying parameters; see Table~\ref{tcvall} for the
critical values corresponding to \( r=1 \)  and \( \alpn=1 \). Their simulation
was performed sequentially on intervals with lengths ranging from 
\(|I_{0}|=m_{0}=10 \) to \( |I_\K|=570 \) observations using a geometric grid
with multiplier \( a=1.25 \), see Section~\ref{Ssprocedure}. (The results are
however not sensitive to the choice of \( a \).)


\linespread{1.0}
\begin{table}[t]
\begin{center}
\caption{Critical values \( \zz_{k} = \zz(|I_{k}|) \)
of the supremum LR test 
for various constant (\( \alpha=\beta=0 \)), ARCH(1) (\( \beta=0 \)), and GARCH(1,1) models; 
\( \omega = 1, r = 1, \alpn = 1 \), and \( \alpha \) and \( \beta \)
are stated in the table.} \label{tcvcomplete} \label{tcvall} 
\begin{tabular}{crcccccccccc}
\hline
\( \zz(|I_{k}|) \)       &         &  \multicolumn{10}{c}{\( \beta \)} \\
\cline{3-12}
\( \alpha \) & \( |I_{k}| \) & ~0.0 & ~0.1 & ~0.2 & ~0.3 & ~0.4 & ~0.5 & ~0.6 & ~0.7 & ~0.8 & ~0.9 \\
\hline    
0.0      &  ~10    & 15.5 & 15.5 & 16.4 & 16.8 & 17.9 & 17.3 & 17.0 & 17.0 & 16.9 & 16.0 \\ 
         &  570    & ~5.5 & ~7.2 & ~7.0 & ~7.0 & ~7.5 & ~7.5 & ~7.4 & ~7.3 & ~7.0 & ~6.7 \\ 
0.1      &  ~10    & 16.3 & 14.5 & 15.1 & 15.9 & 16.4 & 15.9 & 16.1 & 16.0 & 16.0 &      \\
         &  570    & ~8.6 & ~9.0 & ~9.1 & ~9.6 & ~9.8 & 10.7 & 11.5 & 12.5 & 14.0 &      \\
0.2      &  ~10    & 16.7 & 15.2 & 15.7 & 16.2 & 16.9 & 18.9 & 20.1 & 25.1 &      &      \\
         &  570    & ~9.4 & 10.6 & 11.2 & 11.4 & 11.4 & 12.5 & 13.3 & 14.2 &      &      \\
0.3      &  ~10    & 18.5 & 16.4 & 16.7 & 16.9 & 18.1 & 21.8 & 26.4 &      &      &      \\
         &  570    & ~9.7 & 10.8 & 12.0 & 12.4 & 12.9 & 13.5 & 14.5 &      &      &      \\
0.4      &  ~10    & 22.1 & 16.5 & 18.3 & 19.3 & 22.8 & 30.9 &      &      &      &      \\
         &  570    & ~9.9 & 12.0 & 13.0 & 13.4 & 13.9 & 14.7 &      &      &      &      \\
0.5      &  ~10    & 26.2 & 19.1 & 19.5 & 25.4 & 38.1 &      &      &      &      &      \\
         &  570    & 10.7 & 12.6 & 13.8 & 14.0 & 14.6 &      &      &      &      &      \\
0.6      &  ~10    & 33.0 & 22.8 & 25.9 & 32.4 &      &      &      &      &      &      \\
         &  570    & 12.7 & 12.7 & 13.9 & 15.3 &      &      &      &      &      &      \\
0.7      &  ~10    & 41.1 & 24.8 & 29.1 &      &      &      &      &      &      &      \\
         &  570    & 16.8 & 14.7 & 16.1 &      &      &      &      &      &      &      \\
0.8      &  ~10    & 66.2 & 26.4 &      &      &      &      &      &      &      &      \\
         &  570    & 31.5 & 15.8 &      &      &      &      &      &      &      &      \\
0.9      &  ~10    & 88.6 &      &      &      &      &      &      &      &      &      \\
         &  570    & 60.9 &      &      &      &      &      &      &      &      &      \\
\hline
\end{tabular}
\end{center}
\end{table}
\linespread{1.5}

Unfortunately, the critical values depend on the parameters of the underlying (G)ARCH 
model (in contrast to the constant-volatility model). 
They generally seem to increase with the values of the ARCH
and GARCH parameters keeping the other one fixed, see Table~\ref{tcvcomplete}.  
To deal with this dependence on the underlying model parameters,
we propose to choose the largest (most conservative) critical values
corresponding to any estimated parameter in the analyzed data. For example,
if the largest estimated parameters of GARCH(1,1) are \( \hat\alpha=0.3 \)  and
\( \hat\beta=0.8 \), one should use \( \zz(10)=26.4 \)  and \( \zz(570)=14.5 \), 
which are the largest critical values for models with \( \alpha=0.3, \beta \le 0.8 \)
and with \( \alpha \le 0.3, \beta=0.8 \). (The proposed procedure is
however not overly sensitive to this choice as we shall see later.) 

\linespread{1.0}
\begin{table}[t]
\begin{center}
\caption{Critical values \(   \zz(|I_{k}|)   \)  of the supremum LR test 
for some constant volatility, ARCH(1), and GARCH(1,1) models and various values \(   r   \)  and \(   \alpn   \).}\label{tcvalt}
\begin{tabular}{ccccccccccccc}
\hline
\multicolumn{2}{c}{Model \( (\omega,\alpha,\beta) \)} & & 
\multicolumn{2}{c}{\( (0.1,0.0,0.0) \)} & & 
\multicolumn{2}{c}{\( (0.1,0.2,0.0) \)} & & 
\multicolumn{2}{c}{\( (0.1,0.1,0.8) \)} \\
\cline{4-5}
\cline{7-8}
\cline{10-11}
\( r \)  & \( \alpn \)            & & \(\zz(10)\) & \(\zz(570)\)& & \(\zz(10)\) & \(\zz(570)\) & & \(\zz(10)\) & \( \zz(570) \)    \\
\hline
 1.0 & 0.5 & &  16.3     &     ~7.3   & &  17.4     &     11.2    & &  18.7     &     17.1      \\
 1.0 & 1.0 & &  15.4     &     ~5.5   & &  16.7     &     ~9.4    & &  16.0     &     14.0      \\
 1.0 & 1.5 & &  14.9     &     ~4.5   & &  15.9     &     ~8.3    & &  15.2     &     13.4      \\[1ex]
 0.5 & 0.5 & &  10.7     &     ~7.1   & &  11.7     &     10.1    & &  11.7     &     10.1      \\
 0.5 & 1.0 & &  ~8.9     &     ~5.5   & &  10.3     &     ~8.5    & &  10.3     &     ~8.5      \\
 0.5 & 1.5 & &  ~7.7     &     ~4.6   & &  ~9.3     &     ~7.5    & &  ~9.3     &     ~7.5      \\
\hline
\end{tabular}
\end{center}
\end{table}
\linespread{1.5}

Finally, let us have a look at the influence of the tuning constants \( r \)
and \( \alpn \) in (\ref{defcv}) on the critical values for several selected
models (Table~\ref{tcvalt}). The influence is significant, but can be
classified in the following way. Whereas increasing \( \alpn \)  generally
leads to an overall decrease of critical values (cf.\ Theorem~\ref{TCVGARCH}),
but primarily for the longer intervals, increasing \( r \)  leads to an increase of
critical values mainly for the shorter intervals, cf. (\ref{defcv}). In
simulations and real applications, we verified that a fixed choice such
as \( r=1 \)  and \( \alpn=1 \) performs well. To optimize the performance of
the adaptive methods, one can however determine constants \( r \)  and \( \alpn
\) in a data-dependent way as described in Section~\ref{Slambda}. 
We use here this strategy for a
small grid of \( r\in\{0.5,1.0\} \)  and \( \alpn \in \{0.5,1.0,1.5\} \)  and
find globally optimal \( r \)  and \( \alpn \). We will document though
that the differences in the average absolute PE (\ref{eqforerr}) for various
values of \( r \) and \( \rho \) are relatively small.

\begin{figure}
\begin{center}
\includegraphics[scale=0.72]{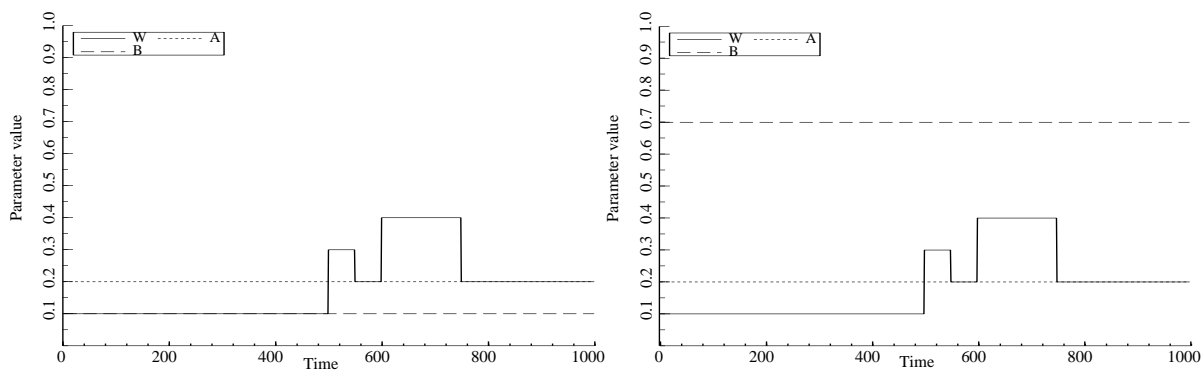}
\caption{GARCH(1,1) parameters of low (upper panel) and high (lower panel) GARCH-effect simulations
for \( t=1,\ldots,1000 \).}\label{fsimpar}
\end{center}
\end{figure}


\subsection{Simulation study}
\label{Sbreaksim}

We aim to examine how well the proposed estimation method is able to adapt to 
long stable (time-homogeneous) periods and to less stable periods with more frequent
volatility changes, and (ii) to see which adaptively estimated model -- local
volatility, local ARCH, or local GARCH -- performs best in different regimes.
To this end, we simulated 100 series from two change-point GARCH models with a
low GARCH effect \( (\omega, 0.2, 0.1) \)  and a high GARCH-effect \( (\omega,
0.2, 0.7) \). Changes in constant \( \omega \)  are spread over a time span of
1000 days, see Figure~\ref{fsimpar}. There is a long stable period at the
beginning (500 days \( \approx \) 2 years) and end (250 days \( \approx \) 1 year) of
time series with several volatility changes between them.

\subsubsection{Low GARCH-effect}

Let us now discuss simulation results from the low GARCH-effect model. First,
we mention the effect of structural changes in time series on the parameter
estimation. Later, we compare the performance of all methods in terms of
absolute PE.

\begin{figure}
\begin{center}
\includegraphics[scale=0.65]{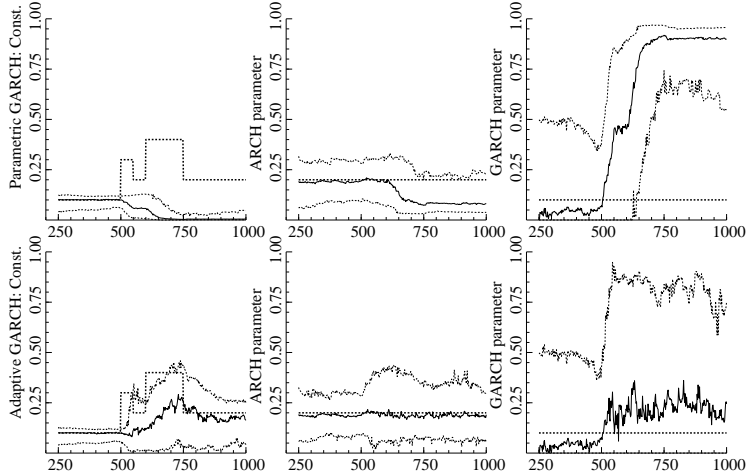}
\caption{The mean (solid line) and 10\% and 90\% quantiles (dotted lines) of
the parameters estimated by the parametric (upper row) and locally adaptive 
(lower row) GARCH methods, \( t=250,\ldots,1000. \) Thick dotted line represents the true parameter value. 
}\label{fparest}
\end{center}
\end{figure}

Estimating a parametric model from data containing a change point will
necessarily lead to various biases in estimation. For example, Hillebrand (2005) 
demonstrates that a change in volatility
level \( \omega \)  within a sample drives the GARCH parameter \( \beta \) very
close to 1. This is confirmed when we analyze the parameter estimates for
parametric and adaptive GARCH at each time point \( t \in [250,1000] \) as
depicted on Figure~\ref{fparest}. The parametric estimates are consistent
before breaks starting at \( t = 500 \), but the GARCH parameter \( \beta \)
becomes inconsistent and converges to 1 once data contain breaks, \( t>500 \).
The locally adaptive estimates are similar to parametric ones before the breaks
and become rather imprecise after the first change point, but they are not too
far from the true value on average and stay consistent (in the sense that the 
confidence interval covers the true values). The low precision of estimation
can be attributed to rather short intervals used for estimation (cf.\  
Figure \ref{fparest} for \( t<500 \)).


Next, we would like to compare the performance of parametric and adaptive estimation
methods by means of absolute PE:
first for the prediction horizon of
one day, \( \cc{H} = \{1\} \), and later for prediction two weeks ahead,
\( \cc{H}=\{1,\ldots,10\} \). To make the results easier to decipher, we present
in what follows PEs averaged over the past month (21 days).
The absolute-PE criterion was also used to determine the optimal values of
parameters \( r \) and \( \alpn \) (jointly across all simulations and for all
\( t=250,\ldots,1000 \)). The results differ for different models: \( r=0.5,
\alpn=0.5 \) for local constant, \( r=0.5,\alpn=1.0 \) for local ARCH, and \(
r=0.5, \alpn=1.5 \) for local GARCH. 

\begin{figure}
\begin{center}
\hspace*{-4mm}\includegraphics[scale=0.55]{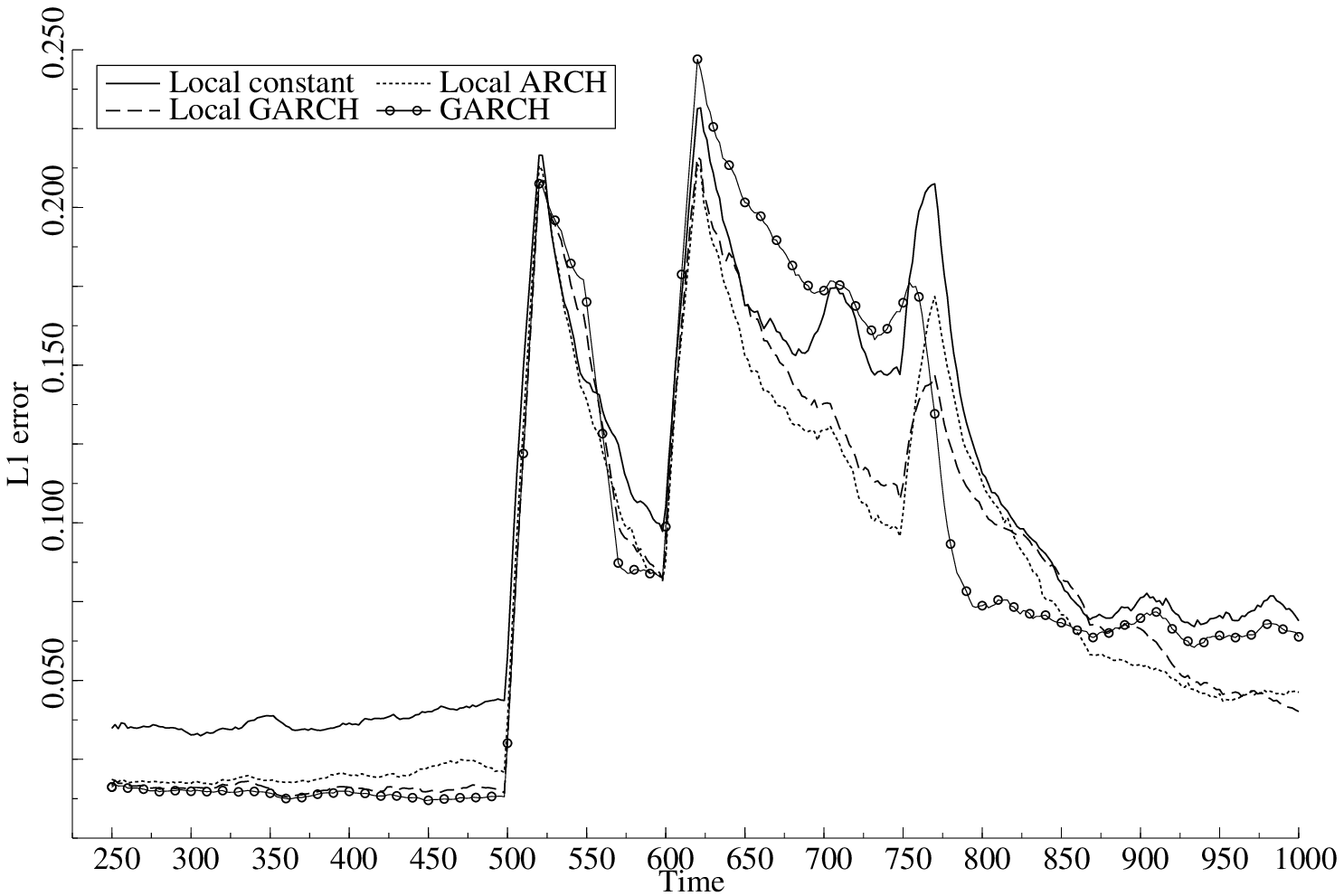}~\includegraphics[scale=0.55]{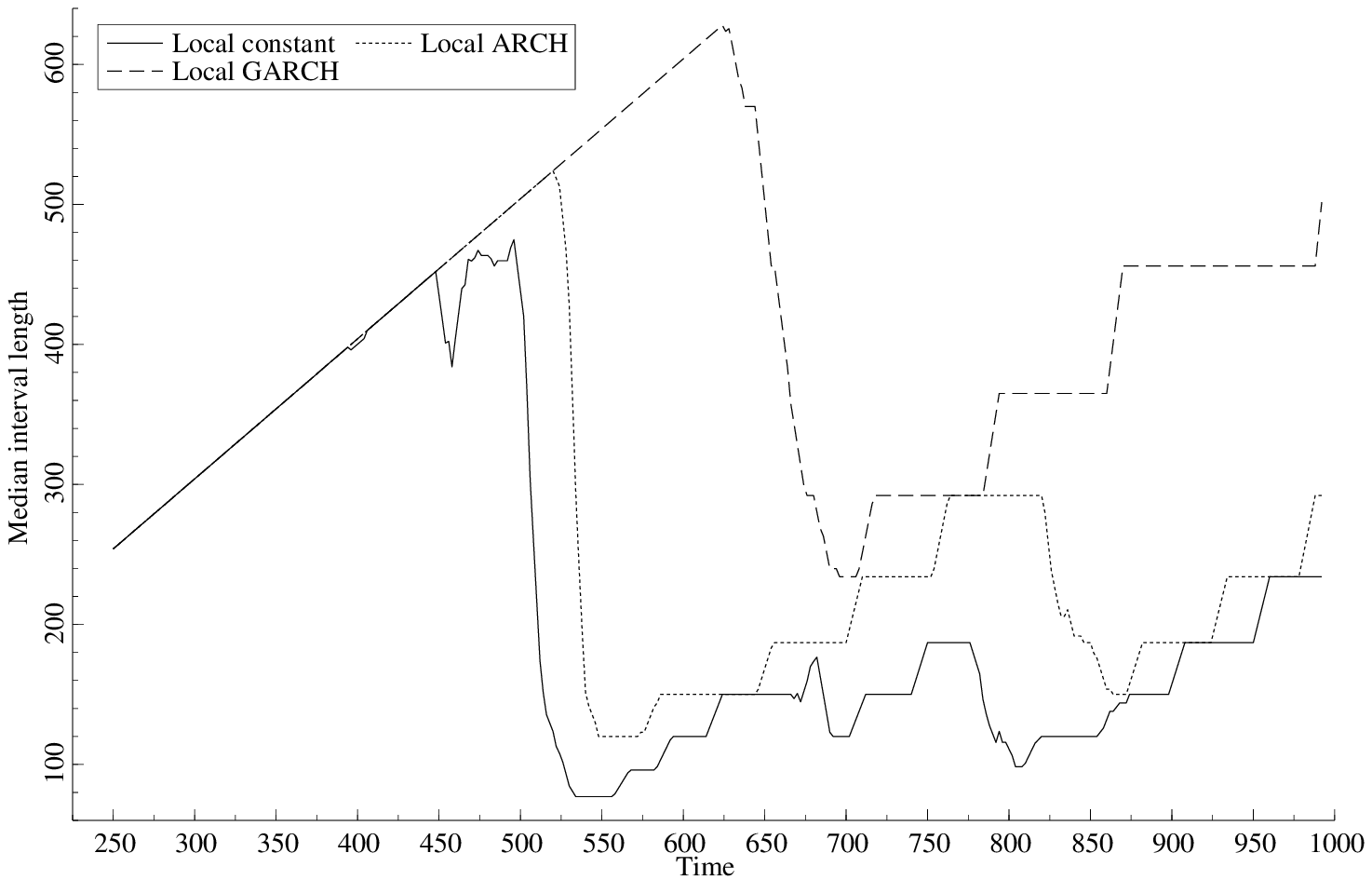}
\caption{\emph{Left panel:} Low GARCH-effect simulations: absolute prediction errors one period ahead averaged
over last month for the parametric GARCH and adaptive local constant,
local ARCH, and local GARCH models; \(   t \in [250,1000]   \).
\emph{Right panel:} The median lengths of adaptively selected intervals for all three pointwise adaptive methods.}\label{fcv1l1err}
\end{center}
\end{figure}

Let us now compare the adaptively estimated local constant, local ARCH, and
local GARCH models with the parametric GARCH, which is the best performing
parametric model in this setup. Forecasting one period ahead, the average PEs
for all methods and the median lengths of the selected time-homogeneous
intervals for adaptive methods are presented on Figure~\ref{fcv1l1err}. First
of all, one can notice that all methods are sensitive to jumps in volatility,
especially to the first one at \( t=500 \): the parametric ones because they ignore
a structural break, the adaptive ones because they use a small amount of data
after a structural change. In general, the local GARCH performs rather
similarly to the parametric GARCH for \( t<650 \) because it uses all historical
data. After initial volatility jumps, the local GARCH however outperforms the
parametric one, \( 650<t<775 \). Following the last jump at \( t=750 \), where
the volatility level returns closer to the initial one, the parametric GARCH is
best of all methods for some time, \( 775<t<850 \), until the adaptive
estimation procedure detects the (last) break, and after it, ``collects'' enough
observations for estimation. Then the local GARCH and local ARCH become
preferable to the parametric model again, \( 850<t \). Interestingly, the local
ARCH approximation performs almost as well as both GARCH methods and even
outperforms them shortly after structural breaks (except for break at \( t=750 \)),
\( 600<t<775 \) and \( 850<t<1000 \). Finally, the local constant volatility is lacking
behind the other two adaptive methods whenever there is a longer time period
without a structural break, but keeps up with them in periods with frequent
volatility changes, \( 500<t<650 \). All these observations can be documented also
by the absolute PE averaged over the whole period \( 250\le t \le 1000 \) (we refer
to it as the global PE from now on): the smallest PE is achieved by local ARCH
(0.075), then by local GARCH (0.079), and the worst result is from local
constant (0.094).

\begin{figure}
\begin{center}
\hspace*{-4mm}\includegraphics[scale=0.55]{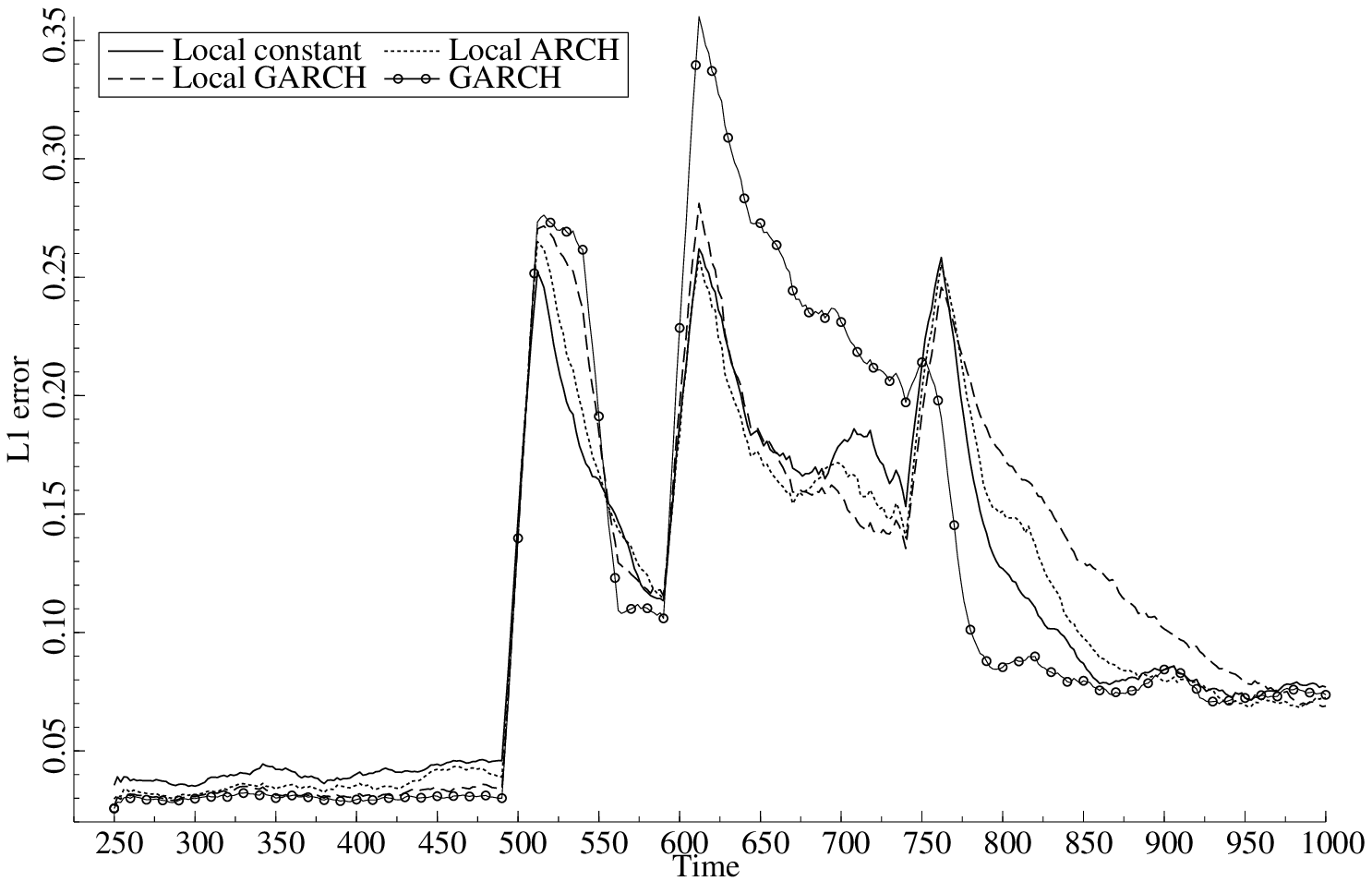}~\includegraphics[scale=0.55]{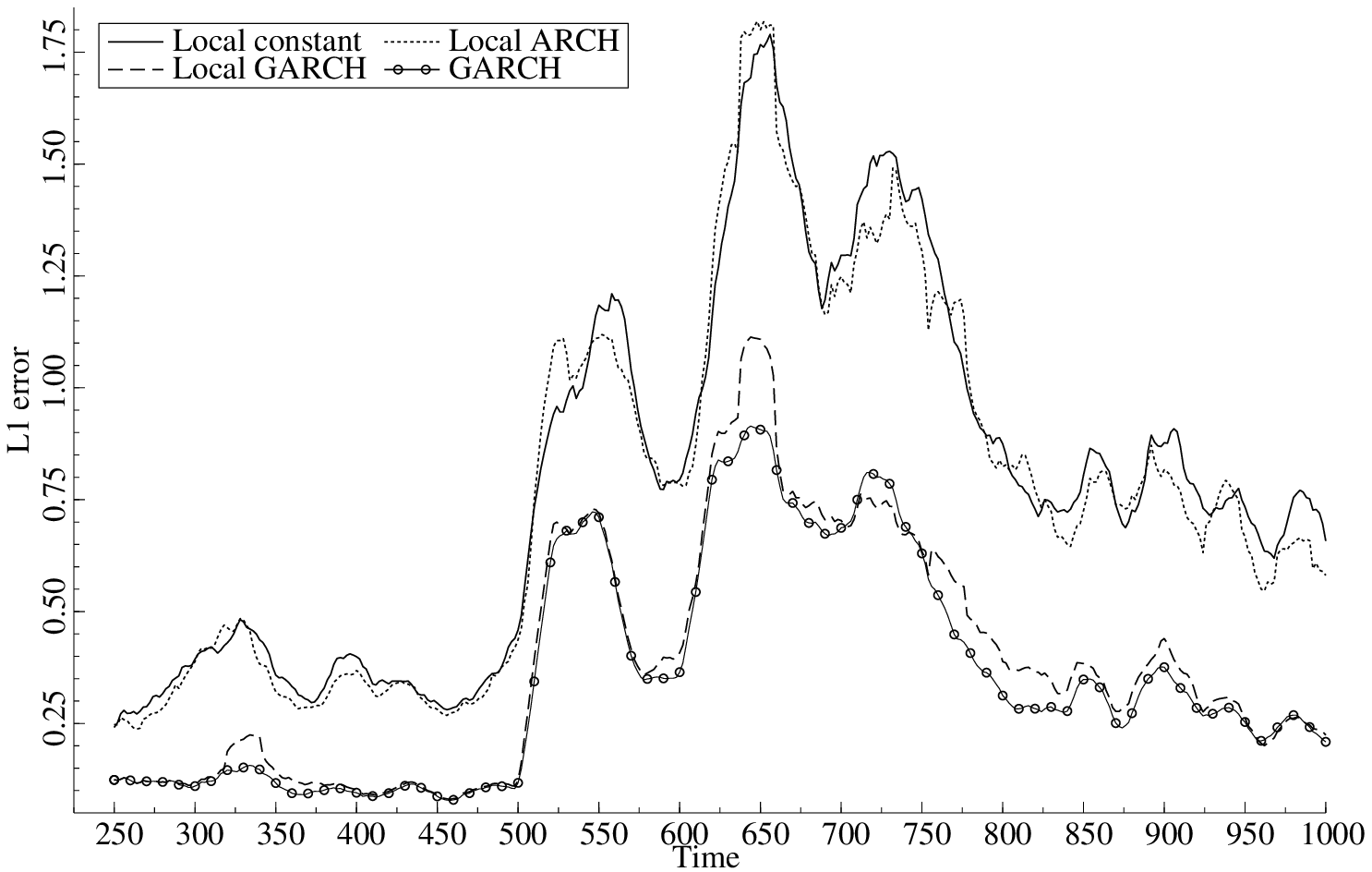}
\caption{\emph{Left panel:} Low GARCH-effect simulations -- absolute prediction errors ten periods ahead averaged
over last month. 
\emph{Right panel:} High GARCH-effect simulations -- absolute prediction errors one period ahead averaged
over last month. In both cases, the parametric GARCH, adaptive local constant,
local ARCH, and local GARCH models are presented for \( t \in [250,1000] \).}
\label{fcv10l1err}
\label{fcv1h1err}
\end{center}
\end{figure}

Additionally, all models are compared using the forecasting horizon of ten
days. Most of the results are the same (e.g., parameter estimates) or similar
(e.g., absolute PE) to forecasting one period ahead due to the
fact that all models rely on at most one past observation. The absolute
PEs averaged over one month are summarized on
Figure~\ref{fcv10l1err}, which reveals that the difference between local constant
volatility, local ARCH, and local GARCH models are smaller in this case. As a
result, it is interesting to note that: (i) the local constant model becomes a
viable alternative to the other methods (it has in fact the smallest global
PE 0.107 from all adaptive methods); and (ii) the local ARCH
model still outperforms the local GARCH (global PEs are 0.108 and 0.116,
respectively) even though the underlying model is GARCH (with a small
value of \( \beta=0.1 \) however).


\subsubsection{High GARCH-effect}

Let us now discuss the high GARCH-effect model. One would expect much more
prevalent behavior of both GARCH models, since the underlying GARCH parameter
is higher and the changes in the volatility level \( \omega \)  are likely to
be small compared to overall volatility fluctuations.
Note that the optimal values of tuning constant \( r \)  and \( \alpn \) differ from 
the low GARCH-effect simulations: \( r=0.5, \alpn=1.5 \) for local constant; 
\( r=0.5,\alpn=1.5 \) for local ARCH; and \( r=1.0, \alpn=0.5 \) for local GARCH.

Comparing the absolute PEs for one-period-ahead forecast at each time
point (Figure~\ref{fcv1h1err}) indicates that the adaptive and parametric
GARCH estimations perform approximately equally well. On the other hand, both
the parametric and adaptively estimated ARCH and constant volatility models
are lacking significantly. Unreported results confirm, similarly to the low
GARCH-effect simulations, that the differences among method are much smaller once
a longer prediction horizon of ten days is used.
%
%



\begin{figure}
\begin{center}
\includegraphics[scale=1.0]{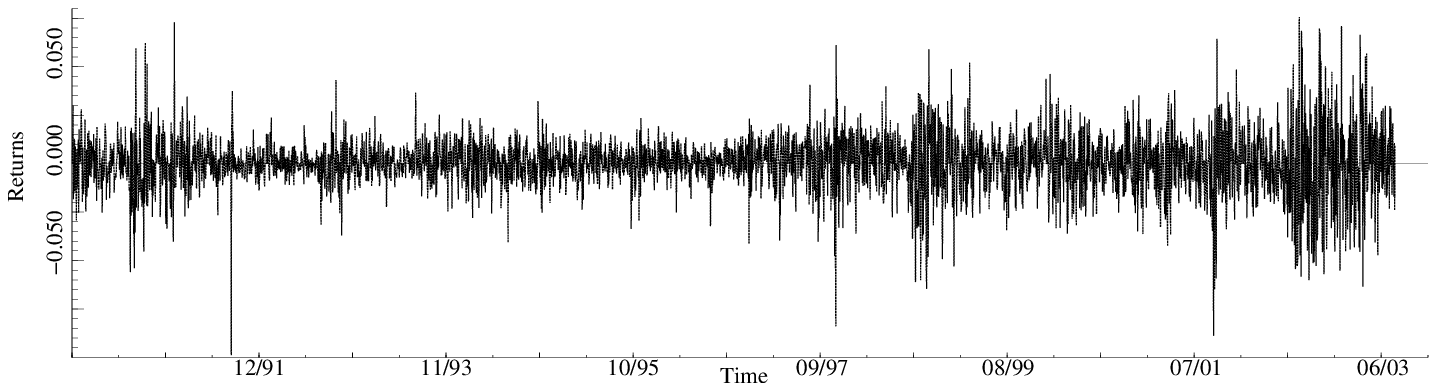}\\
\includegraphics[scale=0.77]{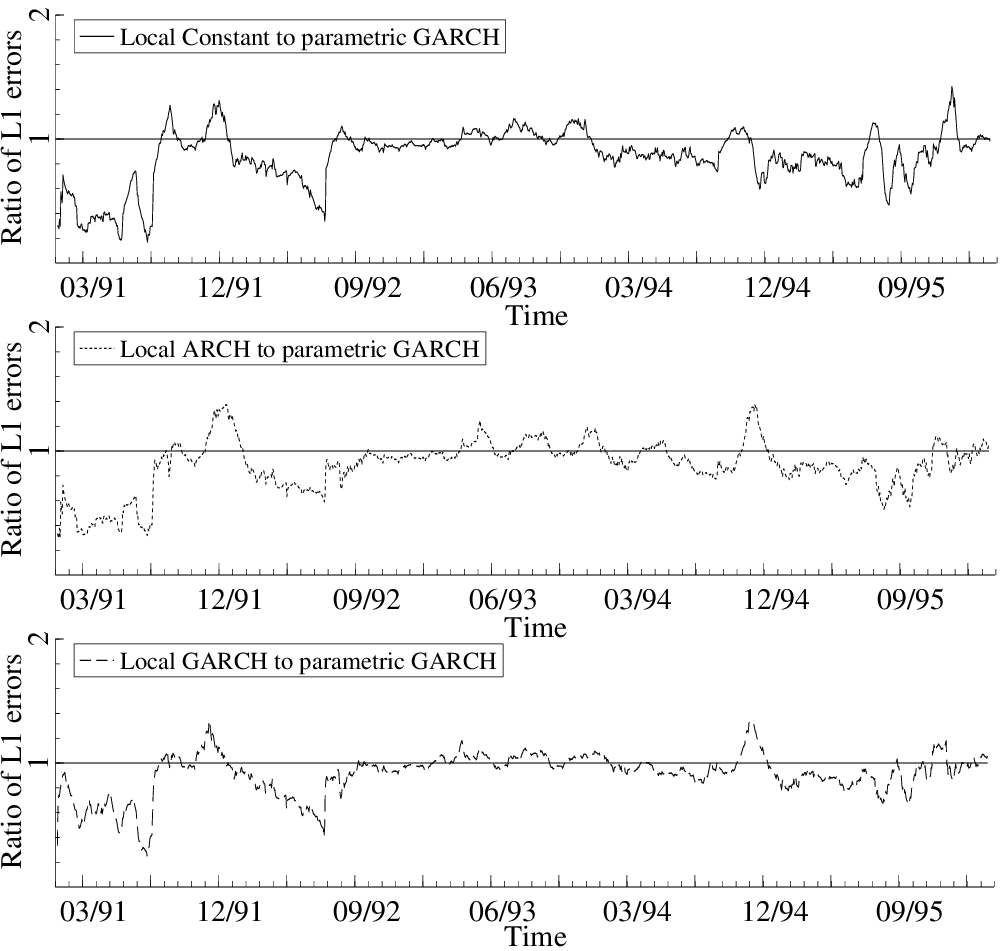}~\includegraphics[scale=0.77]{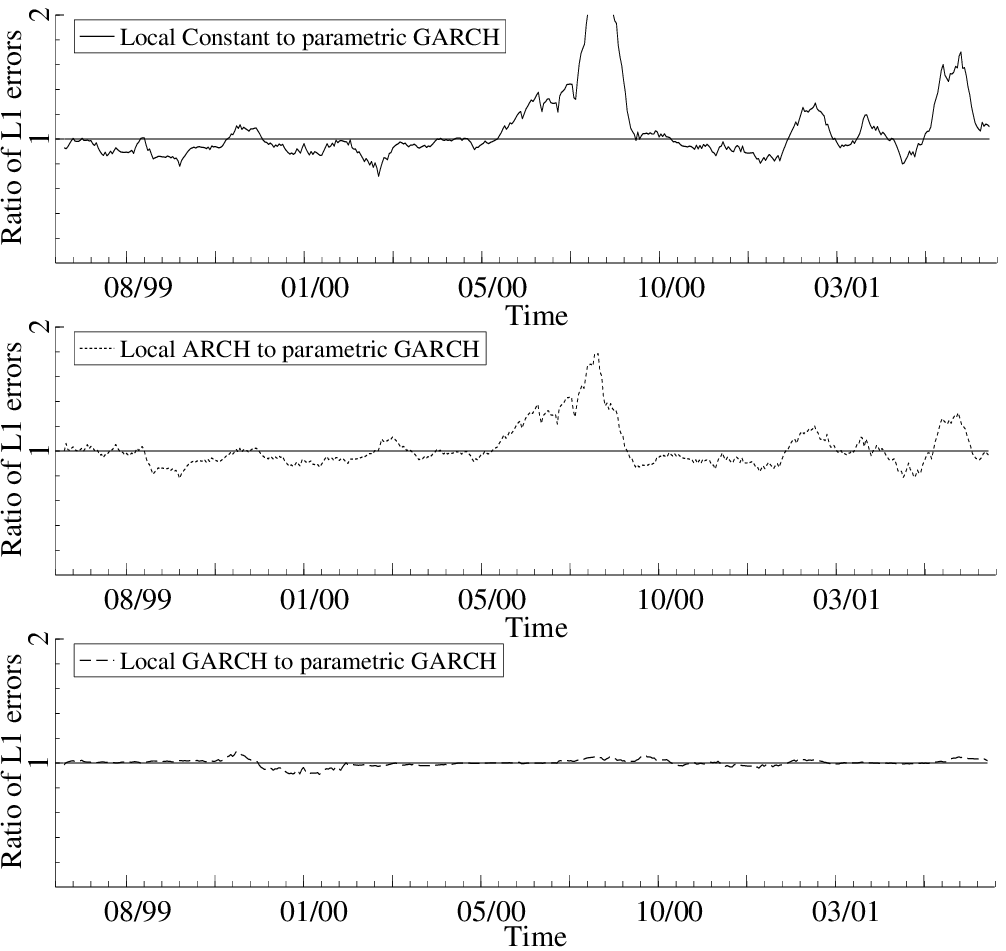}
\caption{\emph{Top panel:} The log-returns of DAX series from January 1990 till December 2002. 
\emph{Bottom panels:} The ratios of the absolute prediction errors of the three 
pointwise adaptive methods to the parametric GARCH for predictions one period ahead 
averaged over one month. The DAX index is considered from January 1992 to March 1997
(left panel) and from July 1999 to June 2001 (right panel).}
\label{dax30-2}
\label{daxreal1992}
\label{daxreal1999}
\end{center}
\end{figure}

\section{Applications}\label{Srealappls}

The proposed adaptive pointwise estimation method will be now applied to real time
series consisting of the log-returns of the DAX and S\&P 500 stock indices
(Sections~\ref{Sdax} and \ref{Ssp500}). 
We will again summarize the results concerning both parametric and adaptive
methods by the absolute PEs one-day ahead averaged over one month. As a
benchmark, we employ the parametric GARCH estimated using last two years of
data (500 observations). Since we however do not have the underlying volatility
process now, it is approximated by squared returns. Despite being noisy, this
approximation is unbiased and provides usually the correct ranking of methods
(Andersen and Bollerslev, 1998). 

\subsection{DAX analysis}\label{Sdax}

Let us now analyze the log-returns of the German stock index DAX from January
1990 till December 2002 depicted at the top of Figure~\ref{dax30-2}. Several
periods interesting for comparing the performance of parametric and adaptive
pointwise estimates are selected since results for the whole period might be hard to
decipher at once.

First, consider the estimation results for years 1991 to 1996. Contrary to
later periods, there are structural breaks practically immediately detected by
all adaptive methods (July 1991 and June 1992; cf. Stapf and Werner, 2003). For
the local GARCH, this differs from less pronounced structural changes discussed
later, which are typically detected only with several months delays. One
additional break detected by all methods occurs in October 1994. Note that
parameters \( r \)  and \( \alpn \)  were \( r=0.5,\alpn=1.5 \) for local
constant, \( r=1.0,\alpn=1.0 \)  for local ARCH, and \( r=0.5,\alpn=1.5 \)  for
local GARCH.

The results for this period are summarized in Figure~\ref{daxreal1992}, which
depicts the PEs of each adaptive method relative to the PEs of parametric
GARCH. First, one can notice that the local constant and local ARCH
approximations are preferable till July 1991, where we have less than 500
observations. After the detection of the structural change in June 1991, all
adaptive methods are shortly worse than the parametric GARCH due to limited
amount of data used, but then outperform the parametric GARCH till the next
structural break in the second half of 1992. A similar behavior can be observed
after the break detected in October 1994, where the local constant and local
ARCH models actually outperform both the parametric and adaptive GARCH. In the
other parts of the data, the performance of all methods is approximately the
same, and even though the adaptive GARCH is overall better than the parametric
one, the most interesting fact is that the adaptively estimated local constant
and local ARCH models perform equally well. In terms of the global PE, the
local constant is best (0.829), followed by the local ARCH (0.844) and local
GARCH (0.869). This closely corresponds to our findings in simulation study
with low GARCH effect in Section~\ref{Sbreaksim}. Note that for other choices
of \( r \) and \( \alpn \), the global PEs are at most 0.835 and
0.851 for the local constant and local ARCH, respectively. This indicates low
sensitivity to the choice of these parameters.


Next, we discuss the estimation results for years 1999 to 2001 (\( r=1.0 \)
for all methods now). After the financial markets were hit by the Asian crisis
in 1997 and Russian crisis in 1998, market headed to a more stable state in
year 1999. The adaptive methods detected the structural breaks in the fall of
1997 and 1998. The local GARCH detected them however with more than one-year
delay -- only during 1999. The results in Figure~\ref{daxreal1999} confirm that
the benefits of the adaptive GARCH are practically negligible compared to the
parametric GARCH in such a case. On the other hand, the local constant and ARCH
methods perform slightly better than both GARCH methods during the first
presented year (July 1999 to June 2000). From July 2000, the situation becomes
just the opposite and the performance of the GARCH models is better (parametric
and adaptive GARCH estimates are practically the same in this period since the
last detected structural change occurred approximately two years ago). Together
with previous results, this opens the question of model selection among
adaptive procedures as different parametric approximations might be preferred
in different time periods. Judging by the global PE, the local ARCH provides
slightly better predictions on average than the local constant and local GARCH --
despite the ``peak'' of the PE ratio in the second half of year 2000 (see
Figure~\ref{daxreal1999}). This however depends on the specific choice of
loss \( \Lambda \) in (\ref{eqforerr}).

Finally, let us mention that the relatively similar behavior of the local constant and local ARCH
methods is probably due to the use of ARCH(1) model, which is not sufficient to capture
more complex time developments. Hence, ARCH(p) might be a more appropriate interim step
between the local constant and GARCH models.


\begin{figure}
\begin{center}
\includegraphics[scale=0.77]{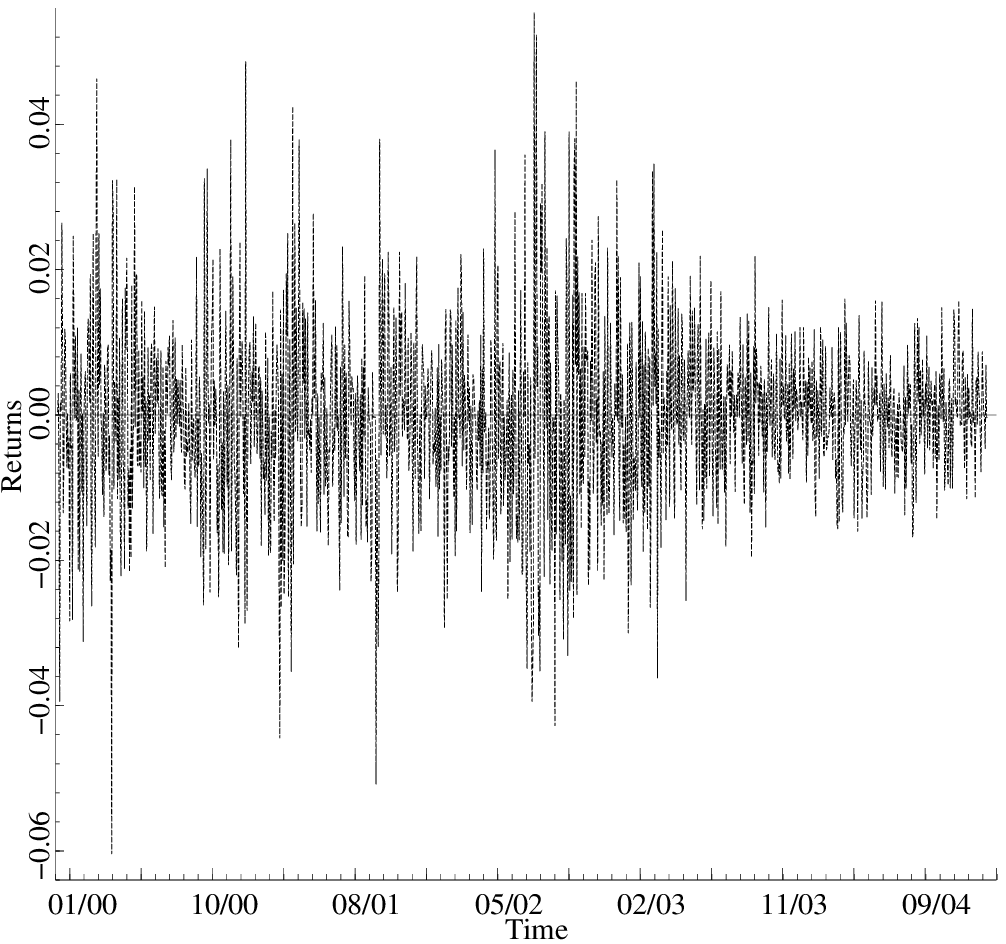}~\includegraphics[scale=0.77]{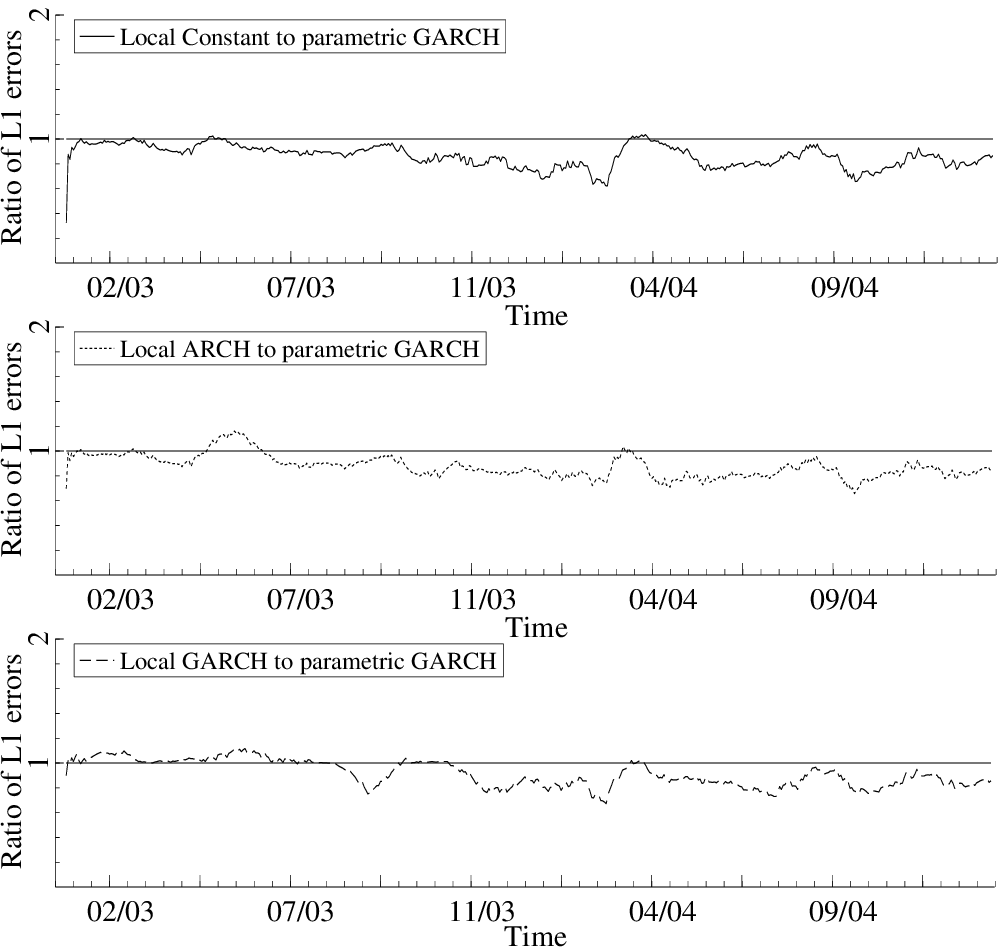}
\caption{\emph{Left panel:} The log-returns of S\&P 500 from January 2000 till December 2004.
\emph{Right panel:} The ratio of the absolute prediction errors of the three pointwise adaptive methods
to the parametric GARCH for predictions one period ahead averaged over one month horizon.
The S\&P 500 index is considered from January, 2003 to December, 2004.}
\label{sp500real20032004}
\label{sp500}
\end{center}
\end{figure}

\subsection{S\&P 500}\label{Ssp500}

Now we turn our attention to more recent data regarding the S\&P
500 stock index considered from January 1990 to December 2004,
%
see Figure~\ref{sp500}. This period is marked by many substantial events affecting
the financial markets, ranging from September 11, 2001, terrorist attacks and
the war in Iraq (2003) to the crash of the technology stock-market bubble
(2000--2002). For the sake of simplicity, a particular time period is again
selected:
year 2003 representing a more volatile period (war in Iraq) and year 2004 being
a less volatile period. All adaptive methods detected rather quickly a
structural break at the beginning of 2003, and additionally, they detected a
structural break in the second half of 2003, although the adaptive GARCH did so
with a delay of more than 8 months. The ratios of monthly PE of all adaptive
methods to those of the parametric GARCH are summarized on
Figure~\ref{sp500real20032004} (\( r=0.5 \)  and \( \alpn=1.5 \) for all
methods).

In the beginning of year 2003, corresponding with 2002 to a more volatile
period (see Figure~\ref{sp500}), all adaptive methods perform as well as the
parametric GARCH. In the middle of year 2003, the local constant and local ARCH
models are able to detect another structural change (possibly less pronounced
than the one at the beginning of 2003 because of its late detection by the
adaptive GARCH). Around this period, the local ARCH shortly performs worse than
the parametric GARCH. From the end of 2003 and in year 2004, all adaptive
methods starts to outperform the parametric GARCH, where the reduction of the
PEs due to the adaptive estimation amounts to 20\% on average. All adaptive 
pointwise estimates exhibit a short period of instability in the first months of
2004, where their performance temporarily worsens to the level of parametric
GARCH. This corresponds to ``uncertainty'' of the adaptive methods about the
length of the interval of homogeneity. After this short period, the performance
of all adaptive methods is comparable, although the local constant performs
overall best of all methods (closely followed by local ARCH) judged by the
global PE.


Similarly to the low GARCH-effect simulations and to the analysis of DAX in
Section~\ref{Sdax}, it seems that the benefit of pointwise adaptive estimation
is most pronounced during periods of stability that follow an unstable period
(i.e., year 2004) rather than during a presumably rapidly changing
environment. The reason is that, despite possible inconsistency of parametric
methods under change points, the adaptive methods tend to have rather large
variance when the intervals of time homogeneity become very short.

\section{Conclusion}

We extend the idea of adaptive pointwise estimation to parametric CH models. In
the specific case of ARCH and GARCH, which represent particularly difficult
cases due to high data demands and dependence of critical values on underlying
parameters, we demonstrate the use and feasibility of the proposed procedure:
on the one hand, the adaptive procedure, which itself depends on a number of
auxiliary parameters, is shown to be rather insensitive to their choice, and on
the other hand, it facilitates the global selection of these parameters by
means of fit or forecasting criteria. The real-data applications highlight the
flexibility of the proposed time-inhomogeneous models since even simple
varying-coefficients models such as constant volatility and ARCH(1) can
outperform standard parametric methods such as GARCH(1,1). Finally, the
relatively small differences among the adaptive estimates based on different
parametric approximations indicate that, in the context of adaptive pointwise
estimation, it is sufficient to concentrate on simpler and less data-intensive
models such as ARCH(\( p \)), \( 0\le p\le 3 \), to achieve good forecasts.



\appendix

\section{Proofs}
%

\begin{proof}[Proof of Corollary \ref{CGARCH11}] Given the choice of \( \zz_\alpha \), it directly
follows from (\ref{thm21eq2}).
\end{proof}

\begin{proof}[Proof of Theorem \ref{TCVGARCH}]
Consider the event \( \cc{B}_{k} = \{ \hat{I} = I_{k-1} \} \) for some 
\( k \le \K \). This particularly means that \( I_{k-1} \) is accepted while 
\( I_{k} = [\TT-m_{k}+1, \TT] \) is rejected; that is, there is 
\( I' = [t',\TT] \subseteq I_{k} \) and \( \tcp \in \Tcp(I_{k}) \) such that
\( T_{I_{k},\tcp} > \zz_{k} = \zz_{I_{k},\Tcp(I_{k})} \).
%
For every fixed \( \tcp \in \Tcp(I_{k}) \) and \( J = I_{k} \setminus [\tcp+1,\TT] \), \( J^{c} = [\tcp+1,\TT] 
\), it holds by definition of \( T_{I_{k},\tcp} \) that
\begin{eqnarray*}
    T_{I_{k},\tcp}
    \le
    L_{J}(\tilde{\thetav}_{J}) + L_{J^{c}}(\tilde{\thetav}_{J^{c}}) 
    - L_{I}(\thetavs)
    =
    L_{J}(\tilde{\thetav}_{J},\thetavs) 
    + L_{J^{c}}(\tilde{\thetav}_{J^{c}},\thetavs) .
\end{eqnarray*}    
This implies by Theorem~\ref{TGARCH11} that
\(
    \P_{\thetavs}\bigl( T_{I_{k},\tcp} > 2\zz \bigr)
    \le 
    \exp\bigl\{ \lexpB(\lambda,\thetavs) - \lambda \zz \bigr\} .
\)
Now, 
\begin{eqnarray*}
    \P_{\thetavs}\bigl( \cc{B}_{k} \bigr)
    \le 
    \sum_{t' = \TT-m_{k}+1}^{\TT-m_{0}}
    \sum_{\tcp = t'+1}^{\TT-m_{0}+1} 
    2\exp\bigl\{ \lexpB(\lambda,\thetavs) - \lambda \zz_{k}/2 \bigr\} 
    \le
    2 \frac{m_{k}^2}{2} \exp\bigl\{ \lexpB(\lambda,\thetavs) - \lambda \zz_{k}/2 \bigr\} .
\end{eqnarray*}    

Next, by the Cauchy-Schwartz inequality
\begin{eqnarray*}
    \E_{\thetavs} \bigl| L_{I_{\K}}(\tilde{\thetav}_{I_{\K}},\hat{\thetav}) \bigr|^{r}
    \!\!=\!
    \sum_{k=1}^{\K} \! \E_{\thetavs} \bigl[ \bigl| 
        L_{I_{\K}}(\tilde{\thetav}_{I_{\K}},\tilde{\thetav}_{k-1}) 
    \bigr|^{r} \bb{1}(\cc{B}_{k}) \bigr]
    \!\!\le\!
    \sum_{k=1}^{\K} \! \E_{\thetavs}^{1/2} \bigl| 
        L_{I_{\K}}(\tilde{\thetav}_{I_{\K}},\tilde{\thetav}_{k-1}) 
    \bigr|^{2r} \P_{\thetavs}^{1/2}(\cc{B}_{k})
\end{eqnarray*}    
Under the conditions of Theorem~\ref{TGARCH11}, it follows similarly to (\ref{THM21EQ3}) that
\begin{eqnarray*}
    \E_{\thetavs} \bigl| L_{I_{\K}}(\tilde{\thetav}_{I_{\K}},\tilde{\thetav}_{k-1}) \bigr|^{2r}
    \le 
    (m_{\K}/m_{k-1})^{2r} \Crlp_{2r}^{*}(\thetavs)
\end{eqnarray*}    
for some constant \( \Crlp_{2r}^{*}(\thetavs) \) and \( k=1,\ldots,\K \), and therefore,
\begin{eqnarray*}
    \E_{\thetavs} \bigl| L_{I_{\K}}(\tilde{\thetav}_{I_{\K}},\hat{\thetav}) \bigr|^{r}
    \le
    [\Crlp_{2r}^*(\thetavs)]^{1/2} \sum_{k=1}^{\K} m_{k} (m_{\K}/m_{k-1})^{r}
    \exp\bigl\{ \lexpB(\lambda,\thetavs)/2 - \lambda \zz_{k}/4 \bigr\}
\end{eqnarray*}    
and the result follows by simple algebra provided that 
\( a_{1} \lambda /4 \ge 1 \) and \( a_{2} \lambda/4 > 2 \).
\end{proof}

\noindent
\emph{Proof of Theorem \ref{Triskbound}.}
The proof is based on the following general result.
\begin{lemma}
\label{Linfbound}
Let \( \P \) and \( \P_{0} \) be two measures such that the Kullback-leibler
divergence \( \E \log (d\P / d\P_{0}) \), satisfies
\( 
    \E \log (d\P / d\P_{0}) \le \Delta < \infty .
\) 
Then for any random variable \( \zeta \) with \( \E_{0} \zeta < \infty \), it holds that
\( 
    \E \log \bigl( 1 + \zeta \bigr)
    \le
    \Delta + \E_{0} \zeta.
\)
\end{lemma}

\begin{proof}
By simple algebra one can check that for any fixed \( y \) the maximum of the function
\( f(x)=xy - x \log x + x \) is attained at \( x = e^{y} \) leading to the
inequality
\( xy \le x \log x - x + e^{y} \).
Using this inequality and the representation
\( \E \log ( 1 + \zeta )
= \E_{0} \{ Z \log ( 1 + \zeta ) \}  \)
with \( Z = d\P/d\P_{0} \) we obtain
\begin{eqnarray*}
    \E \log ( 1 + \zeta )
    &=&
    \E_{0} \{ Z \log ( 1 + \zeta ) \}
     \le 
    \E_{0} ( Z \log Z - Z )
    +
    \E_{0} (1 + \zeta)
    \nn
    &=&
    \E_{0} ( Z \log Z ) + \E_{0} \zeta
    - \E_{0} Z + 1.
\end{eqnarray*}
It remains to note that \( \E_{0} Z = 1 \) and
\( \E_{0} \bigl( Z \log Z \bigr) = \E \log Z  \).
\end{proof}

This lemma applied with
\( \zeta = \losst(\hat{\thetav},\thetav) / \E_{\thetav}\losst(\hat{\thetav},\thetav) \)
yields the result of the theorem in view of 
\begin{eqnarray*}
    \E_{\thetav} \bigl( Z_{I,\thetav} \log Z_{I,\thetav} \bigr)
    &=&
    \E \log Z_{I,\thetav}
    =
    \E \sum_{t \in I} \log  \frac{p[Y_{t},g(X_{t})]}{p[Y_{t},g(X_{t}(\thetav))]}
    \nn
    &=&
    \E \sum_{t \in I}
     \E \Bigl\{ \log  \frac{p[Y_{t},g(X_{t})]}{p[Y_{t},g(X_{t}(\thetav))]} \Big| \cc{F}_{t-1} \Bigr\}
    =
    \E \Delta_{I_{k}}(\thetav). \hspace{33mm} \Box
\end{eqnarray*}

\begin{proof}[Proof of Corollary \ref{CSMBGARCH}]
It is Theorem \ref{Triskbound} formulated for
\( \losst(\thetav',\thetav) = L_{I} (\thetav', \thetav) \).
\end{proof}

\begin{proof}[Proof of Theorem \ref{TpropGARCH}] The first inequality follows from 
Corollary \ref{CSMBGARCH}, the second one from condition (\ref{defcv}) and 
the property \( x \ge \log x \) for \( x>0 \).
\end{proof}

\begin{proof}[Proof of Theorem \ref{TstabGARCH}]
Let \( \hat{k} = k > \ko \). This means that
\( I_{k} \) is not rejected as homogenous. 
Next, we show that for every 
\( k > \ko \) the inequality \( T_{I_{k},\tau} \le T_{I_{k},\Tcp(I_{k})} \le \zz_{k} \)
with \( \tcp = \TT - m_{\ko} = \TT - |I_{\ko}| \) implies 
\( L_{I_{\ko}}(\tilde{\thetav}_{I_{\ko}},\tilde{\thetav}_{I_{k}}) \le \zz_{\ko} \).
Indeed with \( J = I_{k}\setminus I_{\ko} \), this means that, by construction,
\( \zz_{k} \le \zz_{\ko} \) for \( k > \ko \) and 
\begin{eqnarray*}
    \zz_{k}
    \ge 
    T_{I_{k},\tcp}
    =
    L_{I_{\ko}}(\tilde{\thetav}_{I_{\ko}},\tilde{\thetav}_{I_{k}})
    +
    L_{J}(\tilde{\thetav}_{J},\tilde{\thetav}_{I_{k}})
    \ge 
    L_{I_{\ko}}(\tilde{\thetav}_{I_{\ko}},\tilde{\thetav}_{I_{k}}).
\end{eqnarray*}
It remains to note that 
\begin{eqnarray*}
    \bigl| L_{I_{\ko}}(\tilde{\thetav}_{I_{\ko}}, \hat{\thetav}) \bigr|^{r}
    \le 
    \bigl| L_{I_{\ko}}(\tilde{\thetav}_{I_{\ko}}, \hat{\thetav}_{I_{\ko}}) \bigr|^{r}
    \bb{1}(\hat{k} < \ko)
    +
    \zz_{\ko}^{r} \bb{1}(\hat{k} > \ko),
\end{eqnarray*}    
which obviously yields the assertion.
\end{proof}

\end{document}

%% file: expdef.tex
\def\ko{{k^{\circ}}}
\def\lexpB{\mathfrak{e}}
\def\alpn{\rho}
\def\losst{\varrho}

\def\TT{n}

\def\kullb{\cc{K}} 

\def\E{I\!\!E}
\def\P{I\!\!P}

\def\thetav{\bb{\theta}}
\def\thetavs{\bb{\theta}_{0}}

\def\kappa{\varkappa}
\def\zz{\mathfrak{z}}

\def\Crlp{\mathfrak{R}}

\def\K{K}

%% file: mydef.tex
\renewcommand{\Gamma}{\varGamma}
\renewcommand{\Pi}{\varPi}
\renewcommand{\Sigma}{\varSigma}
\renewcommand{\Delta}{\varDelta}
\renewcommand{\Lambda}{\varLambda}
\renewcommand{\Psi}{\varPsi}
\renewcommand{\Phi}{\varPhi}
\renewcommand{\Theta}{\varTheta}
\renewcommand{\Omega}{\varOmega}
\renewcommand{\Xi}{\varXi}
\renewcommand{\Upsilon}{\varUpsilon}
\renewcommand{\(}{$\,}
\renewcommand{\)}{\,$}
\usepackage[mathscr]{eucal}

\newcommand{\cc}[1]{\mathscr{#1}}
\newcommand{\bb}[1]{\boldsymbol{#1}}

\renewcommand{\hat}[1]{\widehat{#1}}
\renewcommand{\tilde}[1]{\widetilde{#1}}
\newcommand{\nn}{\nonumber \\}

\newcommand{\argmax}{\operatornamewithlimits{argmax}}
\def\argmin{\operatornamewithlimits{argmin}}

\def\E{\boldsymbol{E}}
\def\P{\boldsymbol{P}}

\def\T{\top}